\documentclass[journal, 10pt]{IEEEtran}
\usepackage{graphicx}
\DeclareGraphicsExtensions{.pdf,.jpeg,.png}
\usepackage{cite}
\usepackage{amsmath, mathrsfs, amssymb, ntheorem, breqn, mathtools}  
\usepackage{algorithmic}
\usepackage{multicol, array, booktabs, siunitx, enumitem}
\usepackage{url}
\usepackage[linewidth=1pt]{mdframed}
\usepackage[caption=false,font=footnotesize]{subfig}
\def\subparagraph{\paragraph}
\usepackage[compact]{titlesec}
\titlespacing*{\section}{0pt}{*1}{*1}
\titlespacing*{\subsection}{0pt}{*1}{*0}
\titlespacing*{\subsubsection}{0pt}{*1}{*0}

\newcolumntype{L}[1]{>{\raggedright\let\newline\\\arraybackslash\hspace{0pt}}m{#1}}
\newcolumntype{C}[1]{>{\centering\let\newline\\\arraybackslash\hspace{0pt}}m{#1}}
\newcolumntype{R}[1]{>{\raggedleft\let\newline\\\arraybackslash\hspace{0pt}}m{#1}}

\newcommand{\SET}[1]{\ensuremath{\mathcal{#1}}}               
\newcommand{\EV}[1]{\ensuremath{\mathscr{#1}}}                
\newcommand{\VEC}[1]{\ensuremath{\mathbf{\bar{#1}}}}                

\def\EPS{\textit{\ensuremath{\epsilon}-Indistinguishability }}

\def\PRI{\textbf{PRI}}
\def\PRIEPS{\textbf{PRI+}}
\def\PDE{\textbf{PDE}}
\newcommand{\FPRI}[1][k]{\ensuremath{\mathbb{M}}_{#1}}

\def\PD{\ensuremath{\mathbb{D}_{k}}}
\newcommand{\WeePD}[1][k]{\ensuremath{\mathbf{d}}_{#1}}
\def\FPD{\ensuremath{\widehat{\epsilon}_{*,k}}}

\def\FPRIHAT{\ensuremath{\widehat{\mathbb{M}}_{k}}}

\def\DENY{{\text{(\ensuremath{\epsilon, m})--Plausible Deniability }}}
\def\4DENY{{\text{(\ensuremath{4\epsilon, m})--Plausible Deniability }}}

\def\U{\ensuremath{\mathcal{U}}}
\def\S{\ensuremath{\mathcal{S}}}

\def\P{\text{\rm P}}

\def\RPLUS{\ensuremath{\mathbb{R}^{+}}}

\newcommand*{\POL}[1]{{\color[rgb]{0.1,0,0.5}\textbf{#1}}}

\newtheorem{theorem}{Theorem}[section]
\newtheorem{lemma}{Lemma}
\newtheorem{proposition}[theorem]{Proposition}

\newtheorem{definition}{Definition}

\newtheorem{assumption}{Assumption}
\newtheorem*{proof}{Proof}

\newcommand{\QED}{\nobreak \ifvmode \relax \else
      \ifdim\lastskip<1.5em \hskip-\lastskip
      \hskip1.5em plus0em minus0.5em \fi \nobreak
      \vrule height0.75em width0.5em depth0.25em\fi}

\begin{document}
\title{Plausible Deniability in Web Search -- From Detection to Assessment}
\author{P\'{o}l Mac Aonghusa and Douglas J. Leith
\thanks{P. Mac Aonghusa is with IBM Research and Trinity College Dublin.}%
\thanks{D.J. Leith is with Trinity College Dublin.}
}%
\markboth{IEEE Transactions on Information Forensics and Security}{Mac Aonghusa and Leith: Plausible Deniability in Web Search -- From Detection to Assessment}%
\maketitle
\begin{abstract}
\POL{Web personalisation uses what systems know about us to create content targeted at our interests. When unwanted personalisation suggests we are interested in sensitive or embarrassing topics a natural reaction is to deny interest. This is a practical response only if denial of our interest is credible or plausible. Adopting a definition of \emph{plausible deniability} in the usual sense of ``on the balance of probabilities'', we develop a practical and scalable tool called \PDE{} allowing a user to decide when their ability to plausibly deny interest in sensitive topics is compromised.} We show that threats to plausible deniability are readily detectable for all topics tested in an extensive testing program. Of particular concern is observation of threats to deniability of interest in topics related to health and sexual preferences. We show this remains the case when attempting to disrupt search engine learning through noise query injection and click obfuscation. We design a defence technique exploiting uninteresting, proxy topics and show that it provides a more effective defence of plausible deniability in our experiments.
\end{abstract}

\section{Introduction}
\label{sec:intro}
\POL{Encountering inappropriate or unwanted personalised online content can be awkward, depending on social context. What may appear humorous in one situation may be embarrassing, or worse, in another context. When presented with content regarded as inappropriate or discreditable, a user may wish to \emph{deny interest} in the content.} 

\POL{The Oxford English Dictionary defines plausible deniability in terms of reasonable doubt as \emph{``the possibility of denying a fact (especially a discreditable action) without arousing suspicion''}, \cite{oxford}. Informally, user activity observed by the search engine exhibits plausible deniability when user activity is consistent with the user interest in \emph{any} one of several topics, at least one of which is not sensitive for the user, with sufficiently high probability. }

\POL{Accordingly, we assess threats to plausible deniability during web search by testing if content appearing on search result pages can be attributed to user interest in a specific sensitive topic, versus user interest in any other topic, on the balance of probabilities. We ask when can a user \emph{plausibly deny} interest in a range of sensitive topics during online web-search sessions?}

We provide guarantees on the best-possible level of plausible deniability a user can expect during web search in our model. We also introduce a new \emph{Plausible Deniability Estimator}, called \PDE{}, that can be used to assess privacy threats. Outputs from \PDE{} can be represented in terms of readily interpretable probabilities thereby providing an informative indication of risk to the user. 

Our methods are chosen to be straightforward to implement using openly available technologies. We use our results to design and assess counter-measures against threats to plausible deniability during online web-search sessions, using the Google Search as a source of data. We are able to assess threats to plausible deniability from sensitive topic learning in a range of potentially sensitive topics, such as health, finance and sexual orientation. 

Our experimental measurements indicate that, by observing as few as 3-5 revealing queries, a search engine can infer a user is interested in a sensitive topic on the balance of probabilities in $100\%$ of topics tested when no effective defence is provided. In the case of topics related to health and sexual preferences measurements from \PDE{} suggest that the probability a user is interested in sensitive topics related to sexual preference is as high as $90\%$ greater than their probability of interest in any other topic. 

We show that defence strategies based on random query injection of random noise queries and misleading click patterns may provide some protection for individual, isolated queries, but that search engines are able to learn quickly. Significant levels of threat to plausible deniability are detected even when very high levels of random noise are included in the query session or when misleading click patterns are used. These approaches seem to offer little or no improvement to user privacy when considering plausible deniability over the longer term.

In contrast, we find that a defence employing topics that are commercially relevant but uninteresting to the user as \emph{proxy topics} is effective in protecting plausible deniability in the case of $100\%$ of sensitive topics tested. The proxy topic defence differs from traditional obfuscation approaches in actively exploiting the observed ability of the search engine to learn topics quickly, deflecting the focus of interest toward the proxy topic and away from the true topic of interest to the user.

The proxy topic defence works in our experiments, and is simple to apply. However it is important to recognise that we are faced with commercially motivated and increasingly powerful systems with a history of adapting quickly. Our results suggest that search engine capability is continuously evolving so that we can reasonably expect search engines to respond to privacy defences with more sophisticated learning strategies.  Our results also point towards the fact that the text in search queries plays a key role in search engine learning. While perhaps obvious, this observation reinforces the user's need to be circumspect about the queries they ask if they want to avoid search engine learning of their interests. Equally, our results suggest that simple countermeasures, such as proxy topics, that make accurate personalisation more expensive for the online system represent a promising approach in developing new techniques for practical user privacy. 

\section{Related Work}
\label{sec:related:work}
We model a search engine as a black-box by making minimal assumptions about its internal workings. The technique of using predefined \emph{probe queries}, injected at intervals into a stream of true user queries as fixed sampling points, was used in \cite{mac2015don} \POL{where the focus of the paper was detection of possible privacy threats. Extending the idea of probe queries, discussed in \cite{mac2015don}, several new applications are presented in this current paper such as the model of plausible deniability and the associated \PDE{} estimator, the proxy topic defence model and the evaluation of multiple noise and click models for each of these.}
The technique of using predefined probe queries  is borrowed from black-box testing. Modelling an adversary as a black-box, where internal details of recommender systems algorithms and settings are unknown to users, is mentioned in \cite{datta2014privacy} and \cite{Hannak:2013:MPW:2488388.2488435}.

The importance of control over appropriate flow of information is discussed extensively in legal and social science fields. Individual control over personal information flow is discussed in a critique of the \emph{nothing to hide} defence for widespread surveillance in \cite{Solove:2007}. Individual privacy and its social consequences are discussed in \cite{Nissenbaum:2009:PCT:1822585,boyd2011talk}, where agency or control over appropriate disclosure is identified as a key concern.

Plausible deniability as a privacy defence for web search is addressed in the literature. In \cite{arampatzis2013versatile} alternative, less revealing queries are mixed with sensitive topic queries to obfuscate true user interest. In \cite{Arampatzis:2011} queries with generalised terms are used to approximate the search results of a true query, which is never revealed. Plausible deniability for database release has been studied in the context of user data anonymization.  For example, in \cite{DBLP:journals/pvldb/BindschaedlerSG17} a definition of plausible deniability is applied to examine mechanisms for differentially private data set release. More generally, plausible deniability to counteract the impact of personalisation is examined in \cite{cummings2016empirical} for the case of a privacy aware user who knows they are being observed. The authors show that no matter what the behaviour of the user is, it is always compatible with some concern over privacy. In this way the user can offer their awareness of privacy concerns as a general alibi to justify any range of preferences. Plausible deniability for providers of online services is also discussed in the literature. For example, in  \cite{vera2014doccloud} a distributed virtual machine infrastructure is used to provide deniability to online data providers by obfuscating the origin of index data used in recommendations.

The potential of online profiling and personalisation resulting in censorship and discrimination have received growing attention in the research literature. Personalisation as a form of censorship --  termed a filter bubble in \cite{Pariser:2011:FBI:2029079} -- is explored in \cite{Hannak:2013:MPW:2488388.2488435}. In a filter bubble, a user cannot access subsets of information because the recommender system algorithm has decided it is irrelevant for that user. In \cite{Hannak:2013:MPW:2488388.2488435} a filter bubble effect was detected in the case of Google Web Search in a test with 200 users.  Discrimination associated with personalisation has been shown for topics generally regarded as sensitive. In \cite{Sweeney:2013:DOA:2460276.2460278} an extensive review of adverts from Google and \url{Reuters.com} showed a strong correlation between adverts suggestive of an arrest record, and, an individual's ethnicity. In \cite{Guha:2010:CMO:1879141.1879152}, the authors used online advertising targeted {exclusively} to gay men to demonstrate strong profiling in the case of sexual preference. 

Several approaches exist for obfuscating user interactions with search engines with the aim of disrupting online profiling and personalisation.  GooPIR, \cite{Domingo-Ferrer:2009h,SaNchez:2013:KSC:2383079.2383150}, attempts to disguise a user's ``true''  queries by adding masking keywords directly into a true query before submitting to a recommender system. Results are then filtered to extract items that are relevant to the user's original true query.  PWS, \cite{balsa2012ob}, and TrackMeNot, \cite{howe2009trackmenot,peddinti2011limitations}, inject distinct noise queries into the stream of true user queries during a user query session, seeking to achieve an acceptable level of anonymity while not overly upsetting overall utility. Search engine algorithm evolution regarded as a continuous ``arms-race'', is evidenced in the case of Google, for example, by major algorithm changes such as \emph{Caffeine} and \emph{Search+ Your World} have included additional sources of background knowledge from Social Media, improved filtering of content such as \emph{Panda} to counter spam and content manipulation, most recently semantic search capability has been added through \emph{Knowledge Graph} and \emph{HummingBird}, \cite{search_timeline_1}, \cite{search_history}, \cite{search_history_2}.

\POL{Consent to share data for agreed purposes is critical to user trust in service providers and is a key feature of the EU General Data Privacy Regulation (GDPR), \cite{GDPR:2016}. 
Several notable browser add-ons, such as Mozilla Lightbeam, \cite{lightbeam}, and PrivacyBadger, \cite{badger}, facilitate more active user awareness of possible consent issues by helping a user understand where their data is shared with third parties through the sites they visit. XRay, \cite{DBLP:journals/corr/LecuyerDLPPSCG14}, reports high accuracy in identifying which sources of user data such as email or web search history might have {triggered} particular results from online services such as adverts. Active consensual sharing of personal data is investigated in  \cite{Fredrikson:2011} through an in-browser capability, called RePriv, allowing a user to select which portions of their personal data they wish to share with requesters. Both PrivAd, \cite{Guha:2011}, and Adnostic, \cite{Liu:2016} investigate safe profiling through generalisation of user interests in the browser. Both Adnostic and PrivAd  seek to protect the true interests of the user by obfuscating and filtering personalised content through a published interface.}

Evaluation of the effectiveness of privacy defences in the wild was performed by \cite{peddinti2010privacy} in the case of \emph{TrackMeNot} where the authors demonstrate that by using only a short-term history of search queries it is possible to break the privacy guarantees of TrackMeNot using readily available machine-learning classifiers. The importance of background information in user profiling is explored in \cite{Petit2016}. Here a similarity metric measuring distance between \emph{known} background information about a user, given by query history, and subsequent queries is shown to identify 45.3\% of TrackMeNot and 51.6\% of GooPIR queries. \POL{Anti-tracking is an ongoing area of research and recently in \cite{DBLP:conf/ndss/PanCC15} an anti-tracking browser called TrackingFree was reported to be effective at disrupting all of the trackers in the Alexa top-500 list}. Self-regulation has also proven problematic, in \cite{balebako2012measuring}, six different privacy tools, intended to limit advertising due to behavioural profiling, are assessed. The tools assessed implement a variety of tactics including cookie blocking, site blacklisting and Do-Not-Track (DNT) headers. DNT headers were found to be ineffective in tests at protecting against adverts based on user profiling. 

Examples of unsubstantiated and misleading claims by providers of technology to enhance individual privacy are common, \cite{zdnet:charlatans,zdnet:fake}. Concerns about objective evaluation of the claims by providers of such technologies have attracted the attention of Government, where the need for \emph{``Awareness and education of the users \ldots ''} is identified in \cite{enisa:guide} as a key step to building trust and acceptance of privacy technologies for individuals. Accountability and enforcement of accountability for privacy policy is also attracting attention. Regulatory requirements for data handling in industries such as Healthcare (HIPPA) and Finance (GLBA) are well established. The position with respect to handling of data collected by online recommender systems is less clear. In \cite{datta2014privacy}, the author reviews computational approaches to specification and enforcement of privacy policies at large scale.

\POL{Our contribution in this paper is orthogonal to the contributions in the works discussed here. We address the complimentary challenge of privacy monitoring by detecting possible inappropriate use of personal user data by observing personalised outputs. In this respect our approach can be deployed in conjunction with the technologies mentioned.} 

\section{General Setup}
\label{sec:gen:setup}
\subsection{Threat Model}
\label{sec:notation} 

The setup we consider is that of a general user of a commercial, for-profit online search engine. The relationship between the user, denoted \U{}, and the online system, denoted by \S{}, is based on mutual utility where both parties obtain something useful from the interaction -- \U{} gets useful information and recommendations -- while \S{} gets an opportunity to ``up-sell'' to \U{} through targeted content such as advertising. As a commercial business, \S{} recognises cost per user interaction and responsiveness of service are critical to competitiveness. Accordingly content based on user profiling is intended to adapt dynamically to the changing interests of \U{}. \U{} is generally informed regarding  good personal privacy practice and is alert to unwanted or embarrassing personalisation. When \U{} detects threats to her privacy she wishes to assess her ability to plausibly deny her interest in compromising content to avoid awkward social implications. The relationship between \U{} and \S{} is generally described as ``honest but curious''  in the literature. Accordingly we will refer to \S{} as an \emph{observer} rather than the more traditional \emph{adversary}.

Let $\{c_1, \ldots, c_N\}$ denote a set of sensitive categories of interest to \U{}, \emph{e.g.} bankruptcy, cancer, addiction, \emph{etc}. Gather all other uninteresting categories into a catch-all category denoted $c_{0}$. The set $\SET{C} = \{c_0, c_1, \ldots, c_N\}$ is \emph{complete} in the sense that all user topic interests can be represented as subsets of $\SET{C}$ with the usual set operations. We are interested in threats from search engine learning that compromise the ability of \U{} to plausibly deny their interest in sensitive topics. We will assess threats to plausible deniability by testing if content appearing on search result pages can be attributed to interest in a specific topic $c_i \in \SET{C}$, versus interest in any other topic, on the balance of probabilities.  
\begin{figure}
    \centering
        \includegraphics[scale=0.45]{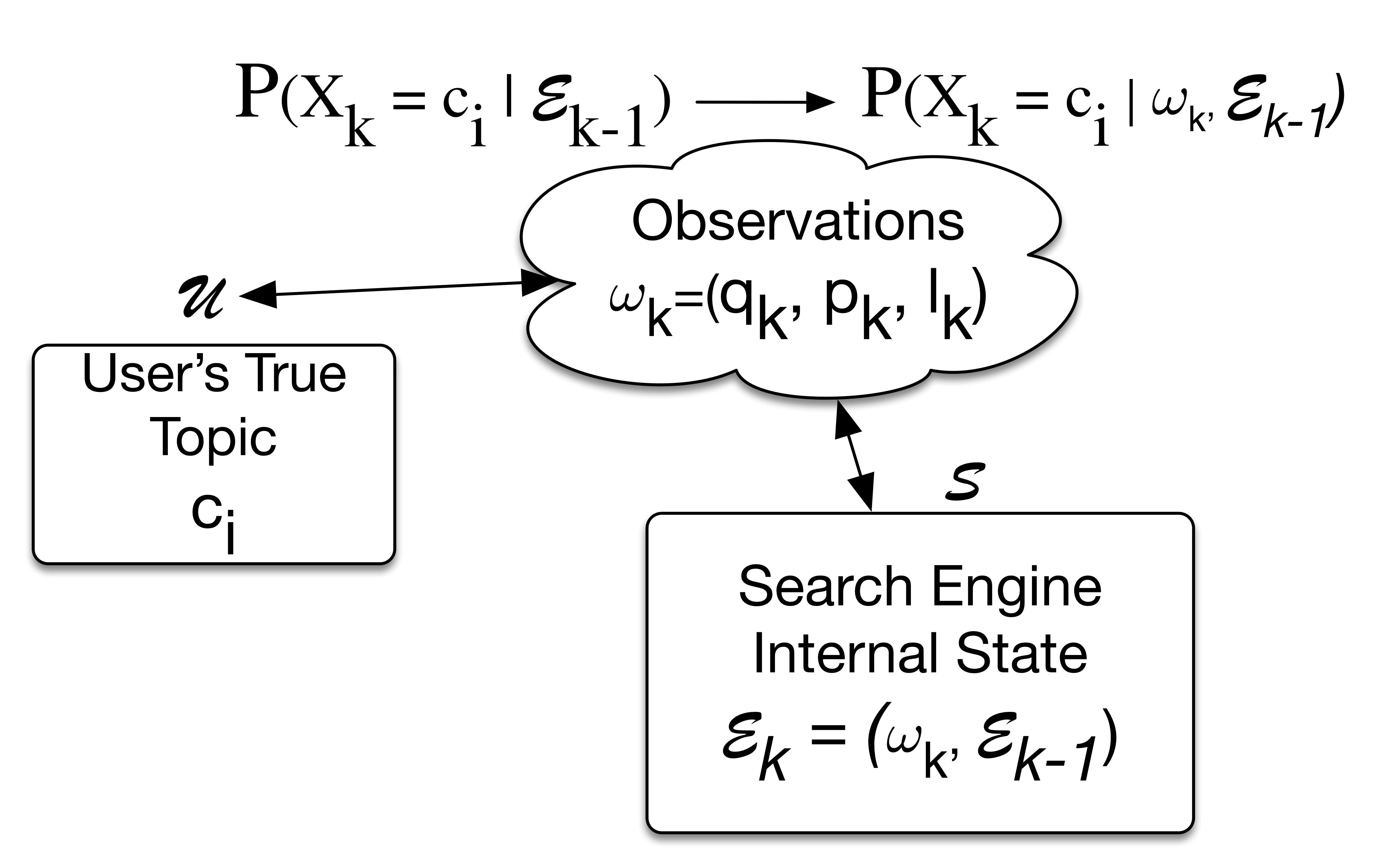}    
   	\caption{Overview of user--search engine interaction\label{fig:overview}}
 \end{figure}\\[-4mm] %

We treat \S{} as a black-box with internal state unknown to \U{}. As a starting point, our initial assumption is that \S{} is motivated to use its  internal state of knowledge of \U{} when producing personalised outputs for \U{}, thereby revealing something about its internal state.
 
\begin{assumption}[Revealing Observations]
	 A search engine selects personalised page content, such as adverts, it believes are aligned with our interests.
	 \label{assumption:1}
\end{assumption}

When a search engine infers that a particular advertising category is likely to be of interest to a user, and it is more likely to generate click through and sales, it is obliged to use this information when selecting which adverts to display. This suggests that, by examining \emph{advert content} recommended by the search engine, it is possible to detect evidence of sensitive topic profiling by the search engine. Assumption~\ref{assumption:1} is fundamental to the application of our approach in that, if \S{} does not produce content that reveals evidence of learning then, since \S{} is a black-box in our model, our approach has nothing to say about observer learning. In summary, we rely on the observer to show his hand through adverts -- our approach can only observe what is shown. In practice this does not appear to be a significant limitation with regard to many topics regarded as sensitive to users. In our experiments we observe an average of $2-3$ adverts per probe query with less than $10\%$ of probe queries resulting in no advert content. We also note that our scope is limited to examining advert content.  We note that in addition to adverts commercial search engines also typically provide additioenal personalised content that could also be tested for evidence of learning, for example Google provides a variety of personalised content such as ``top stories'', related Tweets.  However we leave consideration of these as future work.

A user interacts with a search engine by issuing a query, receiving a web page in response and then clicking on one or more items in the response. In the case of web-search, a single such interaction, labeled with index $j$, consists of a \emph{query, response page, item-click} triple, denoted 
	$\omega_j = \left(q_j, p_j, l_j \right)$. 
	
We model construction of a query $q_j$ as selection of words from a generally available dictionary denoted $\SET{D}$. We assume that words in $\SET{D}$ are matched to topics in $\SET{C}$. The word--topic category matching is not unique and words may be matched to multiple topic categories. 
%
A user session of length $k>0$ steps consists of a {sequence} of $k$ individual steps, and is denoted  $\left\{\omega_k\right\}_{k\ge1}$.  The sequence of interactions $\left\{\omega_k\right\}_{k\ge1}$ is jointly observed by  the user and the search engine -- and perhaps several other third-party observers. 
The relationship between prior and posterior background knowledge at each step $k$ is
\begin{align}
\{\omega_k, \EV{E}_{k-1}\} &= \EV{E}_{k}, \quad k=1,2,\ldots 
\label{eqn:e:k:1} 
\end{align}
\noindent where  $\EV{E}_0$ denotes the initial background knowledge state of \S{} at the beginning of a session immediately before $\omega_1$ is observed. The detail of $\EV{E}_0$ is unknown to \U{} who treats \S{} as a black-box. 
Figure~\ref{fig:overview} illustrates the interaction between user and search engine in our model. 

Let the random variable $\VEC{X}$ with sample space  $\{0,1\}^{N}$ represent user interest in categories in $\SET{C}$ during a session. A value of $1$ in element $i$ of $\VEC{X}$ indicates evidence is detected of user interest in topic $c_i$.   

After each step $k$ of a query session, \S{} can construct a posterior probability distribution for $\VEC{X}$, namely, for $\VEC{x} \in \{0,1\}^{N}$
\begin{align}
\P(\VEC{X} = \VEC{x} \, \vert  \, \Omega_k = \omega_k, \EV{E}_k)
= \P(\VEC{X} = \VEC{x} \, \vert  \, \EV{E}_{k+1})
\label{eqn:actions:001}
\end{align}
\noindent We use $\Omega_k$ to represent the observation at step $k$, so that $\Omega_k = \omega_k$ indicates $\omega_k$ is observed at step $k$.

The individual interest vector for topic $i$ is denoted $\VEC{c}_i$, a vector with a single $1$ in the $i^{th}$ position and $0$ in all other positions. The probability of that a user is interested \emph{only} in topic $i$ at step $k$ of a session and the posterior probability of detecting evidence that a user is interested only in topic $i$ at query step $k$, conditioned on observing $\omega_k$ and background knowledge $\EV{E}_k$, are 
\begin{align*}
	\P(\VEC{X} = \VEC{c}_i \vert \, \EV{E}_k) 
	&\quad\text{and}\quad
	\P(\VEC{X} = \VEC{c}_i \vert \, \Omega_k = \omega_k, \, \EV{E}_{k} )
\end{align*}
\noindent respectively.
Since $\mathcal{C}$ contains all possible topics: 
\begin{align*}
	\sum_{i=0}^{N} \P(\VEC{X} = \VEC{c}_i \vert \, \EV{E}_k) = 1, \quad k=1,2,\ldots
\end{align*}%
\subsection{Example: Single Sensitive Category}
\label{sec:example:bso}
\POL{To illustrate mathematical results as we go, we use a simple ideal model, consisting of a single sensitive category, as an illustrative example. We will refer to it as the Single Sensitive Category (SSC) model. The single sensitive topic is denoted $c^{1}$ and the catch-all, non-sensitive topic representing every other topic that is not part of the sensitive topic is denoted $c^{o}$. }

\POL{Suppose \U{} can issue queries related to either of two topics $\{c^{1}, c^{o} \}$ denoting sensitive and non-sensitive interests respectively. \S{} models the process by which \U{} draws queries according to an initial probability model
\begin{align}
\P(\VEC{X} = c^{1} \vert \EV{E}_0) := p_{0}^{s}, \quad
\P(\VEC{X} = c^{o} \vert \EV{E}_0) := p_{0}^{n}
\label{eqn:transition:bso:0}
\end{align} 
\noindent with  
$p^{1} + p^{0} = 1$.
}
\POL{On observing a query from \U{} at step $k$ the observer \S{} outputs one of $\{\omega^{s}_{k}, \omega^{n}_{k}\}$ according to the associated conditional probabilities at step $k$ given by
\small
\begin{align}
\P(\Omega_{k} = \omega_{k}^{s}\vert \VEC{X} = c^{1}, \EV{E}_{k}) = \P(\Omega_{k} = \omega_{k}^{n}\vert \VEC{X} = c^{o}, \EV{E}_{k}) &:= \pi^{k} \notag \\
\P(\Omega_{k} = \omega_{k}^{n}\vert \VEC{X} = c^{1}, \EV{E}_{k}) = \P(\Omega_{k} = \omega_{k}^{s}\vert \VEC{X} = c^{o}, \EV{E}_{k}) &:= 1-\pi^{k} 
\label{eqn:transition:bso:1}
\end{align}
\normalsize}%
\POL{The SSC model is deliberately simple as the intention is to illustrate mathematical concepts. The model is generally useful for exploring black-box interactions and can be readily extended to include more sophisticated scenarios such as allowing \U{} to select from multiple topics, as would happen when \U{} attempts to obfuscate her interests by switching topics.}%
\subsection{Plausible Deniability}
\label{sec:by:design}
Our threat assessment model is based on \emph{plausible deniability}.   Informally, the user activity observed by the search engine exhibits plausible deniability when, with high probability, is consistent with the user being interested in any one of several topics at least one of which is not sensitive for the user.  That is, the user activity supports reasonable doubt about the user's actual interest in a given sensitive topic. 

In our setup, the topics are $c_i\in\SET{C}$ while the observed activity is $\{\omega_j\}_{j=1}^{k}$ at step $k$ in a session (i.e. the queries, search result pages and associated user clicks).   We formalise plausible deniability as follows.
%
%
\begin{definition}[\DENY{}\label{def:plausible}]
\label{def:plausible:detect:001}
For privacy parameters $\epsilon >0$ and $m>0$ and a set of $N+1$ topics $\SET{C} = \{c_0, \ldots c_N \}$, a user with a true user interest vector $\VEC{x} \in \{0,1\}^{N+1}$ is said to have \emph{\DENY{}} at step $k$ in the query session, if, for observations $\{\omega_{j}\}_{j=1}^{k}$, made at each step $j=1,\ldots k$ of a session by an observer possessing initial background knowledge $\EV{E}_0$ at the beginning of the session, there exist at least $m-1$ other $\VEC{x}_i \in \{0,1\}^{N+1}\setminus\{\VEC{x}\} $ such that
\begin{align}
	e^{-\epsilon} < \PD(\VEC{x}, \VEC{x}_i, \{\omega_{j}\}_{j=1}^{k}) < e^{\epsilon}
	\label{eqn:delta:deny}
\end{align}
\noindent where
\begin{align}
\PD(\VEC{x}, \VEC{x}_i, \{\omega_{j}\}_{j=1}^{k}) = \frac{\P(\Omega_{k}=\omega_{k}, \ldots \Omega_{1}=\omega_{1} \, \vert  \,  \VEC{X} = \VEC{x}, \EV{E}_0 )}
{\P(\Omega_{k}=\omega_{k}, \ldots \Omega_{1}=\omega_{1} \, \vert  \,  \VEC{X} = \VEC{x}_i, \EV{E}_0 )}
\label{eqn:delta:deny:1}
\end{align}
\label{def:delta:deny}
\end{definition}
\noindent For \eqref{eqn:delta:deny} to be well-defined, all probabilities are assumed to be non-zero. In practice, this is not a significant restriction since categories with zero probability are gathered into the catch-all topic $c_0$.
By applying the chain-rule for conditional probability, \eqref{eqn:delta:deny:1} can be rewritten as 
{\small
\begin{align}
&\PD(\VEC{x}, \VEC{x}_i, \{\omega_{j}\}_{j=1}^{k}) = \frac{\P(\Omega_{k}=\omega_{k}, \ldots \Omega_{1}=\omega_{1} \, \vert  \,  \VEC{X} = \VEC{x}, \EV{E}_0 )}
{\P(\Omega_{k}=\omega_{k}, \ldots \Omega_{1}=\omega_{1} \, \vert  \,  \VEC{X} = \VEC{x}_i, \EV{E}_0 )} \notag\\
&= \prod_{j=0}^{k-1} \frac{\P(\Omega_{k-j}=\omega_{k-j} \vert \VEC{X} = \VEC{x}, \Omega_{k-j-1}=\omega_{k-j-1}, \ldots \Omega_{1}=\omega_{1}, \EV{E}_0 )}
{\P(\Omega_{k-j}=\omega_{k-j} \vert \VEC{X} = \VEC{x}_i, \Omega_{k-j-1}=\omega_{k-j-1}, \ldots \Omega_{1}=\omega_{1}, \EV{E}_0 )}
\label{eqn:delta:deny:2} \\
&= \prod_{j=0}^{k-1} \frac{\P(\Omega_{k-j}=\omega_{k-j}\, \vert  \,  \VEC{X} = \VEC{x}, \EV{E}_{k-j} )}
{\P(\Omega_{k-j}=\omega_{k-j} \, \vert  \,  \VEC{X} = \VEC{x}_i, \EV{E}_{k-j} )}
= \prod_{j=0}^{k-1} \WeePD[k-j](\VEC{x}, \VEC{x}_i, \omega_{k-j})
\label{eqn:delta:deny:3}
\end{align}
}%
\noindent where
\begin{align}
\WeePD(\VEC{x}, \VEC{x}_i, \omega_j) := \frac{\P(\Omega_{j}=\omega_{j}\, \vert  \,  \VEC{X} = \VEC{x}, \EV{E}_j )}
{\P(\Omega_{k}=\omega_{j} \, \vert  \,  \VEC{X} = \VEC{x}_i, \EV{E}_j )}
\label{eqn:delta:deny:4}
\end{align}

\noindent is the incremental change in \DENY{} arising from the single observation $\omega_k$ at step $k$ of the session.

\POL{In the case of the SSC model, there are two topics -- sensitive and non-sensitive -- so that \U{} can at best hope for $(\epsilon, m=2)$--Plausible Deniability for the sensitive topic, in which case \eqref{eqn:delta:deny:4} can be written as
\begin{align}
\WeePD(c^{1}, c^{o}, \omega^{s}_j) &= \frac{\P(\Omega_{j}=\omega^{s}_{j}\, \vert  \,  \VEC{X} = c^{1}, \EV{E}_j )}
{\P(\Omega_{k}=\omega^{s}_{j} \, \vert  \,  \VEC{X} = c^{o}, \EV{E}_j )}
&=\frac{\pi^{j}}{1-\pi^{j}}
\label{eqn:deny:bso:1}
\end{align}
\noindent when \S{} emits a sensitive output $\omega^{s}_{j}$ at step $j$, and 
\begin{align}
\WeePD(c^{1}, c^{o}, \omega^{s}_j)
&=\frac{1-\pi^{j}}{\pi^{j}}
\label{eqn:deny:bso:3}
\end{align}
\noindent when a non-sensitive output is emitted at step $j$. Substituting these values into \eqref{eqn:delta:deny:3} and assuming $a$ sensitive outputs and $k-a$ non-sensitive outputs in a session of length $k$, let $k_s$ denote the sub-sequence of steps where sensitive outputs are detected and $k_n = [k] \setminus k_s$ be the steps where non-sensitive outputs are detected so that
\begin{align}
\PD(c^{1}, c^{o}, \{\omega_{j}\}_{j=1}^{k})
= \prod_{i\in k_s}\frac{\pi^{i}}{1-\pi^{i}}\prod_{j\in k_n}\frac{1-\pi^{j}}{\pi^{j}}
\end{align}
 }
\subsection{Comparison with Other Anonymity Measures}
\label{sec:compare:k:dp}
\POL{Intuitively, Definition~\ref{def:plausible:detect:001} is similar to $k$--anonymity in that an observer can only explain observations $\{\omega_{j}\}_{j=1}^{k}$ to within a generalised set consisting of at least $k=m$ topic vectors with probability bounded by the choice of $\epsilon$.  Definition~\ref{def:plausible:detect:001} differs from regular k--anonymity in  requiring both upper and lower bounds in \eqref{eqn:delta:deny:1} since evidence of \emph{loss of interest} in a sensitive topic may be as revealing and potentially embarrassing as evidence of \emph{increase of interest}.}

\POL{Definition~\ref{def:plausible:detect:001} can also be compared with a slightly weaker form of Differential Privacy. Informally, making an observation  should not make \S{} significantly more, or less, confident of user interest in a particular sensitive topic.}

\POL{From \eqref{eqn:delta:deny:4} the incremental change due to a single observation $\Omega_j = \omega_j$ is }
\begin{align}
 \frac{\P(\Omega_{j}=\omega_{j}\, \vert  \,  \VEC{X} = \VEC{x}, \EV{E}_j )}
{\P(\Omega_{k}=\omega_{j} \, \vert  \,  \VEC{X} = \VEC{x}_i, \EV{E}_j )}
= \frac{\P( \VEC{X} = \VEC{x}\, \vert  \, \Omega_{j}=\omega_{j}, \EV{E}_j )}
{\P(\VEC{X} = \VEC{x}_i \, \vert  \,  \Omega_{k}=\omega_{j}, \EV{E}_j )}
\label{eqn:compare:1}
\end{align}
\noindent\POL{by applying Bayes Theorem. Since \eqref{eqn:delta:deny:4} is bounded above and below for at least $m-1$ other $x_i$ when Definition~\ref{def:plausible:detect:001} holds, it follows that  }
\begin{align}
e^{-\epsilon} < \frac{\P( \VEC{X} = \VEC{x}\, \vert  \, \Omega_{j}=\omega_{j}, \EV{E}_j )}
{\P(\VEC{X} = \VEC{x}_i \, \vert  \,  \Omega_{k}=\omega_{j}, \EV{E}_j )} < e^{\epsilon}
\label{eqn:compare:2}
\end{align}
\noindent\POL{for at least $m-1$ other topic vectors $x_i$ -- but not necessarily for \emph{all} topic vectors. In which case we say that m--Differential Privacy holds for $\epsilon > 0$ whenever Definition~\ref{def:plausible:detect:001} holds, meaning that for any topic vector $x$ it is impossible to distinguish it from at least $m-1$ other topic vectors in $\{0,1\}^{N+1}$. This is a slightly weaker statement of Differential Privacy from the usual global definition.}
\subsection{Testing for Plausible Deniability}
\label{sec:p:d:i}
The following \emph{indistinguishability} definition of \emph{privacy risk} measures the change in belief by a search engine due to inference from observed user events relative to its prior belief conditioned on the background data available at the start of the query session. It is adapted from work begun in \cite{mac2015don} and using it allows us to adapt tools originally developed there.
\begin{definition}[\EPS]
\label{def:eps:001} 
For a privacy parameter $\epsilon > 0$, a user with interest vector $\VEC{x} \in \{0,1\}^{N}$ is said to be s said to be $\epsilon$-Indistinguishable with respect to an observation of user actions $\omega_k$ at step $k$, if

\begin{align}
e^{-\epsilon} &\leq \FPRI (\VEC{x}, \omega_k) \leq e^{\epsilon}
\label{def:eps:indist:001}
\end{align}

\noindent where

\begin{align}
\FPRI (\VEC{x}, \omega_k)  &= \frac{\P(\VEC{X} = \VEC{x} \, \vert  \,  \Omega_k = \omega_k, \EV{E}_{k} )}{\P(\VEC{X} = \VEC{x} \, \vert  \,  \EV{E}_{0})} 
\label{eqn:m:u:1} \\
&= \frac{\P(\VEC{X} = \VEC{x} \, \vert  \,  \EV{E}_{k+1} )}{\P(\VEC{X} = \VEC{x} \, \vert  \,  \EV{E}_{0})} \quad\text{(Applying \eqref{eqn:e:k:1})} 
\label{eqn:m:u:2}
\end{align}
\noindent is called the \emph{\EPS score} of the interest vector $\VEC{x}$ for observation $\omega_k$ and background knowledge $\EV{E}_{k}$ at step $k$.
\label{def:eps:indist1}
\end{definition}
In other words, for \EPS to hold at step $k$ of a query session, the conditional posterior distribution should be approximately equal to the prior distribution at the beginning of the query session for the true interests of the user. To ensure \eqref{eqn:m:u:1} is well defined we assume all probabilities in \eqref{eqn:m:u:1} are non-zero, so that $0 < \FPRI (\VEC{x}, \omega_k) < \infty $. Expression \eqref{eqn:m:u:1} implies that if \EPS holds at step $k$ for an interest vector $\VEC{x}$, then $e^{-\epsilon} < (\FPRI[k-1](\VEC{x}, \omega_{k}))^{-1} < e^{\epsilon}$.

The next result provides the necessary connection between \EPS and \DENY to apply tools, developed in \cite{mac2015don} for \EPS, to \DENY. 

\begin{proposition}\label{prop:4eps}
	If \EPS holds on a subset $\SET{I} \subseteq \{0,1\}^{N+1}$ for $\epsilon>0$ for step $k$ and the initial step $1$, then \4DENY holds on $\SET{I}$ for $m \leq \vert \SET{B} \vert$. Furthermore
\begin{align}
	&\PD(\VEC{x}, \VEC{x}_i, \{\omega_{j}\}_{j=1}^{k}) 
	= 	\frac{\FPRI[k](\VEC{x}_1, \omega_{k})}{\FPRI[k](\VEC{x}_2, \omega_{k})} \frac{\FPRI[1](\VEC{x}_2, \omega_{1})}
		{\FPRI[1](\VEC{x}_1, \omega_{1})}
\label{eps:deny:proof:0}
\end{align}
\end{proposition}
\begin{proof}
Assume \EPS holds on $\SET{I} \subseteq \{0,1\}^{N+1}$ then for any $\VEC{x}_1, \VEC{x}_2 \in \SET{I}$. From \eqref{eqn:delta:deny:4} 
\small 
\begin{align}
	&\WeePD(\VEC{x}_1, \VEC{x}_2, \omega_k) = \frac{\P(\Omega_k = \omega_k \vert \VEC{X} = \VEC{x}_1, \EV{E}_k)}{\P(\Omega_k = \omega_k \vert \VEC{X} = \VEC{x}_2, \EV{E}_k)} \notag \\
	&=\frac{\P(\VEC{X} = \VEC{x}_1  \vert \Omega_k = \omega_k, \EV{E}_k)}{\P(\VEC{X} = \VEC{x}_1  \vert \EV{E}_k)}
	\frac{\P(\VEC{X} = \VEC{x}_2  \vert \EV{E}_k)}{\P(\VEC{X} = \VEC{x}_2  \vert \Omega_k = \omega_k, \EV{E}_k)} \label{eps:deny:proof:1} \\
	&\text{(Bayes Theorem)} \notag
\end{align}
\begin{align}
	&=\underbrace{\frac{\P(\VEC{X} = \VEC{x}_1  \vert \Omega_k = \omega_k, \EV{E}_k)}{\P(\VEC{X}=\VEC{x}_1 \vert \EV{E}_0)}}_{(a)}
	\underbrace{\frac{\P(\VEC{X}=\VEC{x}_1 \vert \EV{E}_0)}{\P(\VEC{X} = \VEC{x}_1  \vert \EV{E}_k)}}_{(b)} \times \ldots \notag \\
	&\ldots\times \underbrace{\frac{\P(\VEC{X} = \VEC{x}_2  \vert \EV{E}_k)}{\P(\VEC{X}=\VEC{x}_2 \vert \EV{E}_0)}}_{(c)}
	\underbrace{\frac{\P(\VEC{X}=\VEC{x}_2 \vert \EV{E}_0)}{\P(\VEC{X} = \VEC{x}_2  \vert \Omega_k = \omega_k, \EV{E}_k)}}_{(d)} \label{eps:deny:proof:2}
\end{align}
\begin{align}
	= \frac{\FPRI[k](\VEC{x}_1, \omega_{k})\FPRI[k-1](\VEC{x}_2, \omega_{k-1})}
		{\FPRI[k-1](\VEC{x}_1, \omega_{k-1})\FPRI[k](\VEC{x}_2, \omega_{k})}
		\label{eps:deny:proof:3}
\end{align}%
\normalsize 
\noindent Where expressions (b) and (c) in \eqref{eps:deny:proof:2} are $(\FPRI[k-1]({\VEC{x}_1, \omega_{k-1}}))^{-1}$ and $(\FPRI[k-1](\VEC{x}_2, \omega_{k-1}))^{-1}$ respectively from the definition in \eqref{eqn:m:u:2}. 

Therefore from the definition of $\PD(\VEC{x}, \VEC{x}_i, \{\omega_{j}\}_{j=1}^{k})$ in \eqref{eqn:delta:deny:3}
\begin{align}
	&\PD(\VEC{x}, \VEC{x}_i, \{\omega_{j}\}_{j=1}^{k}) 
	= 	\frac{\FPRI[k](\VEC{x}_1, \omega_{k})}{\FPRI[k](\VEC{x}_2, \omega_{k})} \frac{\FPRI[1](\VEC{x}_2, \omega_{1})}
		{\FPRI[1](\VEC{x}_1, \omega_{1})}
\label{eps:deny:proof:4}
\end{align}
\noindent So that \eqref{eps:deny:proof:0} holds. Since individual elements in \eqref{eps:deny:proof:4} satisfy \EPS for $\epsilon>0$ it follows that \4DENY{} holds as required. 
\QED
\end{proof}%

Proposition~\ref{prop:4eps} provides a basic strategy for asserting when \DENY{} holds. By establishing a value of $\epsilon$ for which a collection of topics $\SET{C} = \{c_0, \ldots c_N\}$ satisfies \EPS, \4DENY follows with, at least, $m = \vert\SET{C}\vert = N+1$. This is a \emph{minimum} guarantee, as there may be topics for which \EPS fails but \4DENY holds. 
%
%
 In our experiments we test whether \U{} can plausibly deny whether or not observed actions can be uniquely associated with interest in a given sensitive topic $c_i$ versus interest in ``any other'' topic in $\SET{C} \setminus \{c_i\}$, so that $m=2$. 
 
For a topic $c_i\in\SET{C}$ the expression for \DENY{} becomes

\small
 \begin{align}
&\PD(\VEC{c}_i, \VEC{c}_{-i} , \{\omega_j\}_{j=1}^k) := 
\frac{\P(\Omega_k = \omega_k, \ldots  \Omega_1 = \omega_1 \vert \VEC{X} = \VEC{c}_i, \EV{E}_0)}{\P(\Omega_k = \omega_k, \ldots  \Omega_1 = \omega_1 \vert \VEC{X} = \VEC{c}_{-i}, \EV{E}_0)}
\label{eqn:discuss:estim:2}
\end{align}
\normalsize
\noindent where $\VEC{c}_{-i}$ denotes the topic interest vector representing interest in the topics $\SET{C}\setminus \{c_i\}$. 

The following result connects $\PD(\VEC{c}_i, \VEC{c}_{-i}, \{\omega_j\}_{j=1}^k)$ to variation in probabilities
\begin{proposition}
\label{prop:estim}
If \DENY{} holds for $\VEC{c}_{i}$ and $\VEC{c}_{-i}$ with $\epsilon>0$ and $m=2$ then
\begin{align}
\vert 	&\P(\Omega_k = \omega_{k}, \ldots \Omega_1=\omega_1 \, \vert  \,  \VEC{X} = \VEC{c}_i, \EV{E}_0 ) \notag\\
&- \P(\Omega_k = \omega_{k}, \ldots \Omega_1=\omega_1  \, \vert  \,  \VEC{X} = \VEC{c}_{-i}, \EV{E}_0 ) \vert \notag \\
	&\geq\left\vert\log\left( \frac{\FPRI[k](\VEC{x}_1, \omega_{k})}{\FPRI[k](\VEC{x}_2, \omega_{k})} \frac{\FPRI[1](\VEC{x}_2, \Omega_{k-1})}
		{\FPRI[1](\VEC{x}_1, \Omega_{k-1})} \right)\right\vert
\label{eqn:prop:estim:2}
\end{align}
And so 
\begin{align}
	\epsilon_{*} := \left\vert\log\left( \frac{\FPRI[k](\VEC{x}_1, \omega_{k})}{\FPRI[k](\VEC{x}_2, \omega_{k})} \frac{\FPRI[1](\VEC{x}_2, \omega_{k-1})}
		{\FPRI[1](\VEC{x}_1, \omega_{k-1})}\right)\right\vert
		\label{eqn:prop:estim:2b}
\end{align}
\noindent is a lower bound for the best possible achievable level of \DENY{}.

\end{proposition}
\begin{proof}
If \DENY{} holds for $\epsilon>0$ then
%
\begin{align}
\left\vert\log( \PD(\VEC{c}_i, \VEC{c}_{-i}, \{\omega_j\}_{j=1}^{k}) )\right\vert < \epsilon
\label{eqn:prop:estim:3}
\end{align}
\noindent so that $\vert\log(\PD(\VEC{c}_i, \VEC{c}_{-i}, \{\omega_j\}_{j=1}^{k}))\vert$ is a lower bound for all $\epsilon >0$ for which \DENY{} holds.
The result follows by substituting the expression in \eqref{eps:deny:proof:3} for $\PD(\VEC{c}_i, \VEC{c}_{-i}, \{\omega_j\}_{j=1}^{k})$ in \eqref{eqn:prop:estim:3}.
\QED  	
\end{proof}%

Proposition~\ref{prop:estim} will be used later to create an estimator for $\epsilon_{*}$ that can be measured in experiments. From now on we simplify our discussion to the case $m=2$ and so experimental results are reported accordingly for the two-topic case, $\{c_i, c_{-i}\}$. 

\section{Implementation}
\label{sec:experiment:setup}

\subsection{Preliminaries}
\label{sec:bayesian:formulation}
During testing we wish to use \eqref{eqn:prop:estim:2b} to create an estimator, we call \PDE{}, to estimate the level of \DENY afforded. Since estimating the quantities in \eqref{eqn:prop:estim:2b} uses \PRI, we recap the bare essentials of \PRI{} and refer the reader to \cite{mac2015don} for more details. 

To test for learning PRI injects a predefined \emph{probe query} into a stream of ``true'' queries during a query session. In this way, any differences detected in advert content in response to probe queries can be compared to identify evidence of learning. An ideal probe should not disrupt the learning process of \S{}. Denote the event a probe query is selected from $\SET{D}$ at step $k+1$ by $\Omega_{k+1} = \omega_{k+1}^{P}$. We formalise the notion of an ideal probe query by demanding that observing $\Omega_{k+1} = \omega_{k+1}^{P}$ should be conditionally independent of the user topic $\VEC{X}$ given the existing background knowledge of the observer

\begin{align}
&\P(\VEC{X} = \VEC{c}_i, \Omega_{k+1} = \omega_{k+1}^{P} \,\vert\,  \EV{E}_k) \notag \\
&= 	
\P(\VEC{X} = \VEC{c}_i, \,\vert\, \EV{E}_k)\cdot \P(\Omega_{k+1} = \omega_{k+1}^{P} \,\vert\,  \EV{E}_k)
\label{eqn:choose:probe:10}
\end{align}
\noindent and so observing the probe query and associated clicks does not provide any more information to the observer about the interests of \U{} than the current background knowledge already provides. From \eqref{eqn:choose:probe:10} 
\begin{align}
&\P(\VEC{X} = \VEC{c}_i \,\vert\, \Omega_{k+1} = \omega_{k+1}^{P}, \EV{E}_k) \notag \\
&= 	
\P(\VEC{X} = \VEC{c}_i, \,\vert\,  \EV{E}_k)
\label{eqn:choose:probe:2}
\end{align}

In practice, choosing an ideal probe query is achieved by selecting words from $\SET{D}$ that match words for several topics in $\SET{C}$ so that it is not possible to associate a single topic in $\SET{C}$ with the probe query.

Construction of \PRI{} is based on several assumptions, the first of these assumptions is that the background knowledge at the first step of a query session, $\EV{E}_1$, provides sufficient description of background knowledge for all subsequent steps of that query session, $\EV{E}_k$.

\begin{assumption}[Sufficiently Informative Responses]\label{a:1}
Let $\SET{K}\subset\{1,2,\cdots\}$ label the sub-sequence of steps at which a probe query is issued.   
At each step $k\in\SET{K}$ at which a probe query is issued,
\begin{align}
\frac{\P(\VEC{X} = \VEC{c}_i | \Omega_k = \omega_k, \EV{E}_{k} )}{\P(\VEC{X} = \VEC{c}_i |  \EV{E}_{1} )}&=\frac{\P(\VEC{X} = \VEC{c}_i | \Omega_k = \omega_k,\EV{E}_{1} )}{\P(\VEC{X} = \VEC{c}_i|  \EV{E}_{1})}
\label{eqn:assumption:informative}
\end{align}
\noindent for each topic ${c}_i \in \SET{C}$.
\end{assumption}

\noindent So that it is not necessary to explicitly use knowledge of the search history during the current session when estimating $\FPRI$ for a topic $c$ as this is already reflected in the search engine response,  $\omega_k$, with the initial background knowledge $\EV{E}_1$ capturing background knowledge up to the start of the session, at step $k$.  Assumption \ref{a:1} greatly simplifies estimation as it means we do not have to take account of the full search history, but requires that the response to a query reveals search engine learning of interest in sensitive category $c$ which has occurred. Assumption~\ref{a:1} was called the ``Informative Probe'' assumption in \cite{mac2015don}.

The next assumption is that adverts are selected to reflect search engine belief in user interests. In this way adverts are assumed to be the principal way in which search engine learning is revealed. Given this assumption, conditional dependence on $\omega_k$ can be replaced with dependence on the adverts appearing on the screen.

\begin{assumption}[Revealing Adverts]\label{a:2} In the search engine response to a query at step $k$ it is the adverts ${a}_k$ on a response page which primarily reveal learning of sensitive categories. 
\begin{align}
\P(\VEC{X} = \VEC{c}_i | \Omega_k = \omega_k,\EV{E}_{1}) = \P(\VEC{X} = \VEC{c}_i | {a}_k,\EV{E}_{1} ), \ k\in \SET{K}\label{eqn:assumption:revealing}
\end{align}
\noindent for each topic ${c}_i \in \SET{C}$.
\end{assumption}

We estimate background knowledge $\EV{E}_1$ by selecting a training data-set, denoted $\SET{T}$, consisting of (label, advert) pairs; where the label is the category in $\SET{C}$ associated with the corresponding advert. For example, when testing for evidence of a single, sensitive topic, called ``Sensitive'', $\SET{T}$ contains items labeled ``Sensitive' or ``Other'', where ``Other'' is the label for the uninteresting, catch-all topic $c_0$. In this way $\SET{T}$ approximates the prior observation evidence available at the start of the query session so that $\SET{T}$ is an estimator for $\EV{E}_1$. 

Text processing of $\SET{T}$ produces a dictionary $\SET{D}$ of \emph{keyword} features.  This processing removes common English language high-frequency words and maps each of the remaining keywords to a stemmed form by removing standard prefixes and suffixes such as ``--ing'' and ``--ed''.  The dictionary $\SET{D}$ represents an estimate of the known universe of keywords according to the background knowledge contained in the training data.

Text appearing in the adverts in a response page is preprocessed in the same way as \SET{T} to produce a sequence of keywords from $\SET{D}$ for each advert; denoted $W=\{w_1,w_2,\cdots,w_{|W|}\}$. Words not appearing in $\SET{D}$ are ignored in our experimental setup for simplicity since sessions are short. In an operational setting it is possible, for example, to update $\SET{D}$ when new keywords are encountered and refactor $\EV{E}_1$ accordingly. 

Let $n_{\SET{D}}(w | W) := \vert \{i:i \in \{1,\cdots,|W|\},w_i = w\} \vert$, denote 
 the number of times an individual keyword $w \in \SET{D}$ occurs in a sequence $W=\{w_1,w_2,\cdots,w_{|W|}\}$. The relative frequency of an individual keyword $w\in\SET{W}$ is therefore,
\begin{align}
\phi_{\SET{D}}(w | W) &= 
\frac{n_{\SET{D}}(w | W)}{\sum_{w\in\SET{D}} n_{\SET{D}}(w | W)}
\label{eqn:phi:define}
\end{align}
\noindent recalling that only keywords $w$ appearing $\SET{D}$ are admissible due to the text preprocessing in our setup.

Let $c_i \in  \SET{C}$ be a sensitive topic of interest,  
and let $\SET{T}(c_i)$ denote the subset of $\SET{T}$ where the labels corresponds to $c_i$. Let $T(\SET{C})$ denote the set of adverts labelled for any topic in $\SET{C}$. The \PRI{} estimator for $\FPRI(\VEC{x}, \omega_k)$ given adverts $a_k$ appearing on the result page for query number $k$, is\footnote{Note that in \cite{mac2015don} the expression given for $\widehat{M}_k (\VEC{c}_i, \omega_k)$ is incorrect and is corrected here.}:

\begin{align}
\widehat{M}_k (\VEC{c}_i, \omega_k) &= \sum_{w\in{\SET{D}}}\left( \frac{\phi_{\SET{D}}(w | \SET{T}(c_i))}{  \phi_{\SET{D}}(w | \SET{T}(\SET{C})) } \cdot {\phi_{\SET{D}}(w|a_k)} \right) \label{eqn:mk001}
\end{align}
where we concatenate all of the advert text on page $k$ into a single sequence of keywords and $\psi_{\SET{D}}(w|a_k)$ is the relative frequency of $w$ within this sequence.  Similarly, concatenating all of the keywords in the training set $ \SET{T}(c_i)$, respectively $ \SET{T}(\SET{C})$, into a single sequence then $\phi_{\SET{D}}(w | \SET{T}(c_i))$, respectively $\phi_{\SET{D}}(w | \SET{T}(\SET{C}))$, is the relative frequency of $w$ within that sequence.
\subsection{Tuning the \PRI{} Estimator}
\label{sec:sparsity}
The quantity $\psi_{\SET{D}}(w|a_k)$ in the expression for the \PRI{} estimator, \eqref{eqn:mk001}, is problematic when the adverts $a_k$ on page $k$ do not contain any of the topic keywords in dictionary $\SET{D}$ i.e. when $a_k=\emptyset$, indicating there is no detectable evidence of a particular topic. To be consistent with the definition of \EPS in Section \ref{def:eps:001}, should result in a \PRI{} score of one for that topic.  We therefore replace $\phi_{\SET{D}}(w|a_k)$ with 
\small
\begin{align}
\psi_{0, \SET{D}}(w | a_k) &= 
\begin{cases}
\phi_{\SET{D}}(w | {a_k})
& \mbox{if } a_k\ne \emptyset \\
1 
& \mbox{otherwise}  
\end{cases} 
\end{align}
\normalsize
Training data is based on a sample of all possible adverts for a particular topic. We may be unlucky so that during the training phase we fail to observe adverts containing infrequently occurring keywords for a particular topic. In this case the relative frequency of such a keyword will be zero and it will not contribute when estimating \PRI{} if encountered in an advert. To address this we introduce a Laplace smoothing parameter $\lambda$ as follows
\small 
\begin{align}
n_{\lambda, \SET{D}}(w | W) &= \lambda + n_{\SET{D}}(w | W) \\
\phi_{\lambda,\SET{D}}(w | W) 
&= \frac{n_{\lambda,\SET{D}}(w | W)}{\sum_{w\in\SET{D}} n_{\lambda, \SET{D}}(w | W)}
\\
\psi_{\lambda, \SET{D}}(w | a_k) &= 
\begin{cases}
\phi_{\lambda,\SET{D}}(w | {a})
& \mbox{if } a_k\ne \emptyset \\
1 
& \mbox{otherwise}  
\end{cases} 
\label{eqn:phi:define2}
\end{align}
\normalsize 
The parameter $0\le \lambda < 1$ enforces a minimum frequency of $1/|\SET{D}|$ on every keyword.  
The expression \eqref{eqn:mk001} is adjusted correspondingly to give a new estimator we call \PRIEPS:

\begin{align}
\FPRIHAT (\VEC{c_i}, \omega_k) &= \sum_{w\in{\SET{D}}}\left( \frac{ \phi_{\lambda, \SET{D}}(w | \SET{T}(c_i))}{ \phi_{\lambda, \SET{D}}(w | \SET{T}(\SET{C})) } \cdot  \psi_{\lambda, \SET{D}}(w | {a_k}) \right) \label{eqn:mk001:2}
\end{align}
We will use the \PRIEPS{} estimator, given by \eqref{eqn:mk001:2}, from now on in this paper, unless stated otherwise. In our experiments we find empirically, through verification with the training data, that choosing the parameter $\lambda = 0.001$ worked well.

\subsection{The \PDE{} Estimator}
\label{sec:dpe}
Substituting the \PRIEPS{} estimator $\widehat{\FPRI{}}$  for $\FPRI{}$ in \eqref{eqn:prop:estim:2b} gives the \PDE{} estimator
\small
\begin{align}
	\FPD{} &= \left\vert\log\left(\frac{\widehat{\FPRI[k]}(\VEC{c}_i, \omega_{k})}{\widehat{\FPRI[k]}(\VEC{c}_{-i}, \omega_{k})}\frac{\widehat{\FPRI[1]}(\VEC{c}_{-i}, \omega_{1})}
		{\widehat{\FPRI[1]}(\VEC{c}_i, \omega_{1})}\right)\right\vert
		\label{eqn:dpe:estim:1}
\end{align}
\normalsize 
From Proposition~\ref{prop:estim}, the \PDE{} estimator in \eqref{eqn:dpe:estim:1} can be interpreted directly as the best possible level of \DENY{} a user can claim in the case $m=2$. We report the maximum value of \PDE{} measured by probe step in our experiments to show the worst possible \DENY{} scenario for \U{}. We also report the median value of \PDE{} as a representative bound for approximately $50\%$ of the samples. An example of reporting is shown in Table~\ref{tbl:pde:example:measured} for the reference topic ``gay''.%
\begin{table}[h]
\captionsetup{position=bottom,skip=0pt,belowskip=0pt}
\caption{Measured \FPD for Reference Topic versus Any Other Topic, reported as ``max (median)'', by Probe Query Sequence\label{tbl:pde:example:measured}}
\setlength{\tabcolsep}{0.35cm}
\begin{scriptsize}

\begin{tabular}{@{}lccccc@{}}
\textbf{\shortstack[l]{Reference\\Topic\\{}}}   & \textbf{Probe 1}  & \textbf{Probe 2}   & \textbf{Probe 3}  & \textbf{Probe 4} & \textbf{Probe 5}  \\ 
\cmidrule[0.5pt]{1-6}
\textbf{gay} & 64 (33) & 47 ( 5) & 72 (25) & 48 (25) & 48 (19) \\
 \cmidrule[0.5pt]{1-6}                
 \end{tabular}\\[-4mm]         
              
\end{scriptsize}
\end{table}%
For example, from Table~\ref{tbl:pde:example:measured}, a reported maximum value of \PDE{} of $47\%$ in the second column indicates that the difference in probabilities that \U{} is uniquely interest in the reference topic versus being interested in any other topic is \emph{at least} $47\%$ in the worst case by probe step $5$. The median value of $25\%$ in parentheses in the Probe 3 and 4  columns indicates that the difference in probabilities can be expected to be at least $25\%$ in $50\%$ of cases by probes $3$ and $4$. Overall the results suggest that \DENY{} is unlikely to constitute a reasonable defence in this case. 

Reported values of \PDE{} may increase, or decrease, during a session as individual queries are judged as more, or less, revealing by the \PDE{} estimator. Inspection of the query scripts generated for the topic $c_i=\text{Gay}$, for example, shows that the queries associated with probe step $3$ are \emph{same sex relationships} and \emph{how do i know if I'm gay}, both of which appear revealing. The queries from the test script corresponding to probe steps $4$ and $5$ are \emph{HIV symptoms}, \emph{HIV treatment}, \emph{HIV men} and \emph{aids men} which may not point as distinctly to specific interest in the $c_i=\text{Gay}$ as they could reasonably be associated with health concerns.  

The zeroth probe in a session is always run first, before any other query, to establish a baseline \PRIEPS{} score for the session. As a result the measured \PDE{} values for the zeroth probe is always $0$ for both maximum and median values and is not reported in our results.

One popular approach to designing defences of \DENY{} is to attempt to \emph{hide in the crowd}. For example, by injecting varying degrees of noise in the stream of observations $\{\omega_j\}$ in the hope that \S{} will not detect the true sub-stream of sensitive events. 
In \cite{mac2015don}, the authors observe that varying click patterns is seen to change the absolute volume of adverts appearing on a page. As both user clicks and queries are potential indicators of user interest for an observer we test injected noise from both queries and clicks as possible defence strategies.

An alternative tactic is to invert the previous approach by instead attempting to \emph{hide in plain sight}. By choosing a non-sensitive \emph{proxy topic}, chosen to attract personalised content \U{} can then carefully hide true, sensitive queries in a stream of proxy topic queries.
By demonstrating clear interest in a \emph{proxy} non-sensitive topic \U{} may tip the balance of probability toward the proxy topic by drawing the attention of the observer \S{}.  
%


\section{Experimental Results}
\label{sec:experiment:results}
\subsection{Preliminaries}
\label{sec:topical:queries}
To facilitate easy comparison we use the same experimental data collection setup as \cite{mac2015don}. We summarise the key elements here with additional detail in the Appendix and refer the reader to \cite{mac2015don} for full details. 

User interest topic categories taken from \cite{mac2015don}, are used in our experiments. Of the user interest topics, (i) ten are sensitive categories associated with subjects generally identified as causes of discrimination (medical condition, sexual orientation \emph{etc}) or sensitive personal conditions (gambling addiction, financial problems \emph{etc}), (ii) a further sensitive  topic is  related to London as a specific destination location, providing an obviously interesting yet potentially sensitive topic that a recommender system might track, (iii) the last topic is a catch-all category labeled ``Other''. 

To construct sequences of queries for use in test sessions, we select a \emph{probe query}, providing a predefined sampling point for data collection. Numbering the probes in a session starting from $0$, the zeroth query issued in every session is a probe query. The zeroth probe is used to establish the baseline for calculations of the \PDE{} estimator for subsequent probe queries. The \PDE{} estimator, from \eqref{eqn:dpe:estim:1}, of the zeroth probe in a session is $0$ and so is not included in reports of experimental results. Measurements of \PDE{} values are reported for each of the probe queries $1$--$5$ during experiments providing a consistent sample for analysis.  

In our experiments, when implementing the ``Proxy Topic'' defence model, we choose three uninteresting, proxy topics likely to attract adverts, namely \emph{tickets for music concerts}, searching for \emph{bargain vacations} and \emph{buying a new car}. 

All scripts were run for $3$ registered users and $1$ anonymous user on the Google search engine, yielding a data set consisting of $21,861$ probe queries in total across all of the test user interest topics. Test data was divided  into individual test data sets based on different test configurations with each test data set consisting of approximately $1,000$ probe queries. 

A separate hold-back was created for a common training data set of approximately $1,000$ queries. The \PDE{} estimator in \eqref{eqn:dpe:estim:1} uses the training data-set to model the prior background knowledge $\EV{E}_0$. We do not re-train \PDE{} during testing as new adverts are encountered. Experimental measurements of \PDE{} are with respect to the common training set for consistent comparison.

All queries in a test session were automatically labelled with the intended topic of the test session as given by the query script used. For example, all queries from a session about ``prostate`` are labeled as ``prostate'' including probe queries.  In this respect the labels capture intended behaviour of queries, rather than attempting an individual interpretation of specific query keywords during a user session. Test data is automatically divided into $7$ folds for processing so that, reported statistics are taken over $7$ distinct, randomised sub-samples of test data. 
%

Before proceeding to testing with \PDE, we verify \PRIEPS{} by comparing its detection capability with previous results obtained in \cite{mac2015don} for the \PRI{} estimator \POL{and compare the performance of \PRIEPS{} with alternative implementations using Naive Bayes and Support Vector Machine as sensitive topic detectors}. 

Comparison results between \PRI{} and \PRIEPS{} are shown in  Table~\ref{tbl:old:new}. and were produced by processing data taken from \cite{mac2015don} but applying the \PRIEPS{} estimator to decide which topic is detected. For comparison with  \cite{mac2015don}, we declare a topic $c_i$ has been detected during a query session, consisting of $5$ probe queries, if \emph{at least one} of the $5$ probe queries is detected as topic $c_i$. For comparison, detection results for the \PRI{} estimator from Table XIV(b) in \cite{mac2015don}, are reproduced as Table~\ref{tbl:old:new}(b). The True Detection rates using \PRIEPS{} estimator are better or equal for each topic than the rates reported in \cite{mac2015don}. The False Detection rates are also better or equal in the case of all topics tested comparing favourably with the results obtained in \cite{mac2015don}.

\POL{Comparison of \PRIEPS{} with alternative implementations was performed by taking results from Multinomial Naive Bayes (NB) and Linear SVM (SVM) classifiers to estimate the probabilities in the definition of $\FPRI$ in \eqref{eqn:m:u:2}. The intent of the comparison is to determine which of the NB, \PRIEPS{} and SVM estimators detect privacy threats, using the definition of $\FPRI$ in \eqref{eqn:m:u:2}, for test items previously labeled as ``sensitive'' by examining the topic of the query used. To qualify as a privacy threat we choose a value of $e^{\epsilon} > 1.1$. We expect precision to be substantially less than 100\% for all estimators because the threshold will filter out weaker detections where $1.0 < e^{\epsilon} \leq 1.1$.}
\begin{figure}[h]
    \centering   
        \includegraphics[scale=0.55]{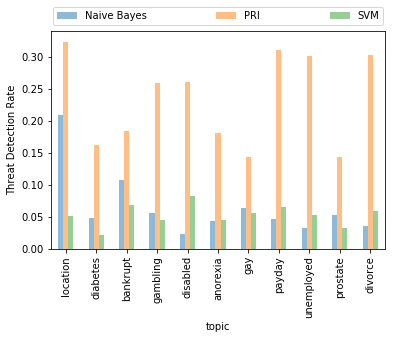} 
        \caption{Comparison of Naive Bayes, \PRIEPS{} and Support Vector Machine estimators. (as Threat Detection Rate by Topic)\label{fig:compare:nb:pri:svm}}   
 \end{figure}%
\POL{Other than varying how $\FPRI$ was estimated, all other inputs and calculations were identical. A common test data set was constructed by selecting 5,500 result pages for each sensitive topic and then randomly selecting an additional 5,500 result pages labeled for the non-sensitive topic.  In this way each sensitive topic had a balanced verification data set of 11,000 labeled items. Each verification data-set was divided randomly into $20\%-80\%$ test--training sets and calculations repeated 5 times for 5-fold verification of each of the NB, \PRIEPS{} and SVM estimators. The Multinomial Naive Bayes and Linear SVC modules from the Python Sklearn package were used to construct the NB and SVM estimators, \cite{scikit-learn}.}
\POL{After common preprocessing each of the NB, \PRIEPS{} and SVM classifiers were trained and probability estimates captured for the 5-fold test data-sets. A threat is declared ``detected'' if the calculated valued of $\FPRI$ for the sensitive topic exceeds $1.0$. Precision of sensitive topic threat detection is shown by topic in Figure~\ref{fig:compare:nb:pri:svm} for the NB, \PRIEPS{} and SVM approaches.} 

\POL{The results Figure~\ref{fig:compare:nb:pri:svm} indicate that that the \PRIEPS{} estimator detects significantly more true-positive detection results than either of the NB or SVM estimators for all sensitive topics tested. The initial detection sensitivity of each of these estimators is influenced by the labelling assigned to examples in the training set. We adopt the perspective that privacy tools should err on the side of caution so that high detection sensitivity in the initial ``out of the box'' stage is a prudent approach. In a real-world application of \PRIEPS{} the user would provide incremental training examples over time reflecting their tolerance of privacy risk and so tune \PRIEPS{}. }
\subsection{Establishing a Baseline} 
\label{sec:test:setup}
We begin with a sequences of queries, interleaved with probe queries, in what we term a ``no click, no noise'' model. Here there is no injected noise and no items are clicked on any of the search results pages. This model provides a baseline, where the queries alone are available to the recommender to learn about a user session as it progresses. Measurements of \PDE{} for all topics using the ``no click, no noise'' model are shown in Table~\ref{tbl:base:noclick:nonoise}. 
\begin{table}[h]
\captionsetup{position=bottom,skip=0pt,belowskip=0pt}
\caption{Measured \FPD for Reference Topic versus Any Other Topic, reported as ``max (median)'', by Probe Query Sequence\label{tbl:base:noclick:nonoise}}
\setlength{\tabcolsep}{0.30cm}
\begin{scriptsize}

\subfloat[No Click, No Noise]{

\begin{tabular}{@{}lccccc@{}}
\textbf{\shortstack[l]{Reference\\Topic\\{}}}   & \textbf{Probe 1}  & \textbf{Probe 2}   & \textbf{Probe 3}  & \textbf{Probe 4} & \textbf{Probe 5}  \\ 
\cmidrule[0.5pt]{1-6}
\textbf{anorexia} & 56 (52) & 56 (52) & 56 (52) & 56 (52) & 56 (52) \\
\textbf{bankrupt} &  1 ( 1) & 55 (43) & 55 (39) & 58 (48) & 56 (48) \\
\textbf{diabetes} & 40 (38) & 40 (38) & 40 (38) & 40 (38) & 40 (38) \\
\textbf{disabled} &  9 ( 9) &  9 ( 9) &  9 ( 9) & 40 (40) & 40 (33) \\
\textbf{divorce} & 41 (31) & 75 (65) & 56 (46) & 79 (68) &  79 (68) \\
\textbf{gambling} & 16 (12) & 18 (16) & 66 ( 4) & 57 (17) & 18 ( 3) \\
\textbf{gay} & 64 (33) & 47 ( 5) & 72 (25) & 48 (25) & 48 (19) \\
\textbf{location} & 10 ( 2) & 11 ( 3) & 11 (10) & 18 ( 7) & 18 ( 9) \\
\textbf{payday} &  2 ( 2) &  2 ( 2) & 21 ( 2) &  2 ( 2) &  2 ( 2) \\
\textbf{prostate} & 52 (17) & 52 (17) & 52 (17) & 52 (17) & 52 (17) \\
\textbf{unemployed} &  7 ( 5) &  7 ( 6) &  7 ( 6) & 13 ( 7) &  7 ( 7) \\
 \cmidrule[0.5pt]{1-6}                
 \end{tabular}        
}\\[-4mm]              
\end{scriptsize}
\end{table}%
For the health-related topics $\text{Anorexia, Diabetes, Prostate, Bankrupt, Divorced, Gay}$ the reported results are high, indicating lack of  plausible deniability for each of these topics. 
 It is concerning that personal circumstances, health status and sexual orientation appear to be the most revealing topics according to our experiments. In the case of the topic $\text{Disabled}$ there is more cause of concern about \DENY{} as the session progresses. On inspection of the associated query script this appears to be again related to the specificity of the queries at each probe step. At the beginning of this script the queries are related to availability of services -- for example, \emph{locations of disabled parking} -- while later queries are more specific to named conditions -- for example, \emph{treatment for spina bifida}. 

The topics $\{Location, Payday, Unemployed\}$ appear among the topics of least concern from the perspective of \DENY{}. Both of the topics Payday and Unemployed asked queries about availability of social support services whereas queries for the topic Bankrupt asked about availability of paid professional services such as lawyers and accountants. It is perhaps an illustration of the motivations of a for-profit service where users seeking social supports are of less interest than users seeking expensive paid services. 

Overall, measurements of \PDE{} in experiments appear to agree with expectations from inspection of the underlying queries.  
 Our results suggest that queries are a strong signal to the observer of user interest, and that estimates from \PDE{} appear to distinguish queries that are strongly revealing of specific topic interest from more generic queries where plausible deniability is clearer. 

\subsection{The Effect of Random Noise Injection}
\label{sec:testing:noise}

Following from Section~\ref{sec:by:design}, we now consider the impact of injecting non-informative queries chosen at random from our popular query list into a user session. We simply refer to these as ``random noise'' queries. We consider three levels of random noise queries for testing purposes:

\begin{description}[style=unboxed,leftmargin=0cm]
	\item [``Low Noise''] The automation scripts select uninteresting queries uniformly at random from the  top-query list and inject a single random noise query after every topic-specific query so that the ``signal-to-noise ratio'' of sensitive to noise queries in this case is $1:1$.
	\item [``Medium Noise''] Here the automation scripts inject two randomly selected  queries after each topic-specific query for a signal to noise ration of $1:2$.
	\item [``High Noise''] In this noise-model with the highest noise setting, three random noise queries are injected, resulting in a signal-to-noise ratio of $1:3$.
\end{description} 

Note also that the automation scripts were configured to ensure the relevant number of noise queries was always injected \emph{immediately before} each probe query.   Our intention was to construct a ``worst case'' for detection of learning, where probe queries are always separated from sensitive user queries by the specified number of noise queries.

\begin{table}[h]
\captionsetup{position=bottom,skip=0pt,belowskip=0pt}
\caption{Measured \FPD for Reference Topic versus Any Other Topic, reported as ``max (median)'', by Probe Query Sequence\label{tbl:base:noclick:allnoise}}
\setlength{\tabcolsep}{0.30cm}
\begin{scriptsize}

\subfloat[No Click, Low Noise]{

\begin{tabular}{@{}lccccc@{}}
\textbf{\shortstack[l]{Reference\\Topic\\{}}}   & \textbf{Probe 1}  & \textbf{Probe 2}   & \textbf{Probe 3}  & \textbf{Probe 4} & \textbf{Probe 5}  \\ 
\cmidrule[0.5pt]{1-6}
\textbf{anorexia} & 54 (45) & 54 (45) & 54 (45) & 54 (45) & 54 (45) \\
\textbf{bankrupt} & 16 ( 9) & 56 (50) & 52 (39) & 54 (45) & 56 (45) \\
\textbf{diabetes} & 46 (35) & 46 (35) & 46 (35) & 46 (35) & 46 (35) \\
\textbf{disabled} &  9 ( 3) &  9 ( 8) &  9 ( 7) & 33 ( 7) & 40 (32) \\
\textbf{divorce} & 13 ( 7) & 123 ( 8) & 54 ( 8) & 85 ( 6) &  85 ( 6) \\
\textbf{gambling} & 18 (16) & 18 (16) & 52 (18) & 18 (10) & 18 (18) \\
\textbf{gay} & 73 (61) & 73 (70) & 76 (46) & 79 (74) & 79 (70) \\
\textbf{location} & 18 (16) & 18 (10) & 18 (10) & 18 (10) & 18 (10) \\
\textbf{payday} &  3 ( 2) &  3 ( 2) &  4 ( 3) &  4 ( 3) &  4 ( 3) \\
\textbf{prostate} & 21 (16) & 21 (16) & 21 (16) & 21 (16) & 21 (16) \\
\textbf{unemployed} &  7 ( 3) &  7 ( 3) & 13 ( 9) & 13 ( 9) & 13 ( 9) \\
\cmidrule[0.5pt]{1-6}                
\end{tabular}        
} \\[-3mm]
\subfloat[No Click, Med Noise]{

\begin{tabular}{@{}lccccc@{}}
\textbf{\shortstack[l]{Reference\\Topic\\{}}}   & \textbf{Probe 1}  & \textbf{Probe 2}   & \textbf{Probe 3}  & \textbf{Probe 4} & \textbf{Probe 5}  \\ 
\cmidrule[0.5pt]{1-6}
\textbf{anorexia} & 55 (53) & 53 (53) & 53 (53) & 53 (53) & 53 (53) \\
\textbf{bankrupt} & 11 ( 8) & 48 (33) & 51 (43) & 52 (38) & 52 (38) \\
\textbf{diabetes} & 38 (38) & 38 (38) & 38 (38) &  38 (38) &  38 (38) \\
\textbf{disabled} &  4 ( 4) &  8 ( 7) &  1 ( 1) & 40 (36) & 40 (36) \\
\textbf{divorce} & 19 ( 9) & 65 (31) & 44 (31) & 72 (50) &  72 (50) \\
\textbf{gambling} & 18 (16) & 18 (17) & 18 (18) & 31 ( 3) & 18 (10) \\
\textbf{gay} & 89 (68) & 89 (69) & 88 (64) & 93 (73) & 93 (64) \\
\textbf{location} & 18 (10) & 18 (10) & 18 ( 7) & 18 ( 7) & 10 ( 7) \\
\textbf{payday} &  6 ( 3) &  6 ( 3) &  6 ( 3) &  6 ( 2) &  6 ( 1) \\
\textbf{prostate} & 32 (14) & 32 (14) & 18 (13) & 18 (13) & 18 (13) \\
\textbf{unemployed} & 13 ( 5) & 13 (10) & 13 ( 7) & 13 ( 9) &  7 ( 4) \\
 \cmidrule[0.5pt]{1-6}                
 \end{tabular}        
} \\[-3mm]
\subfloat[No Click, High Noise]{

\begin{tabular}{@{}lccccc@{}}
\textbf{\shortstack[l]{Reference\\Topic\\{}}}   & \textbf{Probe 1}  & \textbf{Probe 2}   & \textbf{Probe 3}  & \textbf{Probe 4} & \textbf{Probe 5}  \\ 
\cmidrule[0.5pt]{1-6}
\textbf{anorexia} & 48 (48) & 48 (48) & 48 (48) & 48 (48) & 48 (48) \\
\textbf{bankrupt} & 16 (10) & 65 (51) & 65 (48) & 65 (49) & 65 (49) \\
\textbf{diabetes} & 41 (38) & 41 (38) & 41 (38) & 41 (38) & 41 (38) \\
\textbf{disabled} &  9 ( 9) &  9 ( 9) &  9 ( 5) &  9 ( 7) &  9 ( 8) \\
\textbf{divorce} & 41 (27) & 75 (38) & 56 (22) & 75 (29) &  75 (29) \\
\textbf{gambling} & 21 (16) & 21 ( 3) & 21 ( 4) & 29 (16) & 18 ( 4) \\
\textbf{gay} & 86 (64) & 86 (64) & 80 (43) & 94 (59) & 94 (59) \\
\textbf{location} & 10 (10) &  8 ( 8) &  8 ( 8) & 18 (13) & 18 (13) \\
\textbf{payday} &  3 ( 2) &  4 ( 2) &  4 ( 2) &  4 ( 2) &  3 ( 1) \\
\textbf{prostate} & 17 (15) & 17 (15) & 17 (15) & 17 (15) & 17 (15) \\
\textbf{unemployed} & 10 ( 7) & 13 ( 7) & 13 ( 7) & 13 ( 7) & 13 ( 7) \\
 \cmidrule[0.5pt]{1-6}                
 \end{tabular}        
} \\[-6mm]                          
\end{scriptsize}
\end{table}%

Table~\ref{tbl:base:noclick:allnoise}(a-c) shows the measured \PDE{} values for Low, Medium and High levels of noise respectively for the ``no click'' model. The \PDE{} values for all levels of noise are similar to the ``no click, no noise'' baseline values in Table~\ref{tbl:base:noclick:nonoise}. 

Overall, there is no consistent reduction in values across all topics for all noise levels, indicating that injecting random noise queries does not have a consistent effect. In some cases, such as topic $\text{Gay}$, measured values of \PDE{} increase for all noise levels indicating that noise injection \emph{worsens} the user's ability to assert \DENY{}.

These results indicate that even the ``High Noise'' model fails to reduce the measured values of \PDE{} in a coherent way, so that injecting random noise has not improved plausible deniability significantly with any consistency. We conclude that injection of random noise, even at substantial levels, is not observed to provide a useful defence for plausible deniability in our experiments. 

\subsection{The Effect of Click Strategies}
\label{sec:testing:clicks}

We now consider whether it is possible to disrupt search engine learning by careful clicking of the links on response pages.  
Intuitively, from the search engine's point of view, clicking on links is a form of active feedback by a user and so potentially informative of user interests.  This is especially true when, for example, a user is carrying out exploratory search where their choice of keywords is not yet well-tuned to their topic of interest.  Previous studies have also indicated that there is good reason to believe that user clicks on links are an important input into recommender system learning. In \cite{mac2015don} (Section $6.4$), user clicks emulated using the ``Click Relevant'' click-model were reported to result in increases of $60\%$ -- $450\%$ in the advert \emph{content}, depending on the ``Sensitive'  topic tested.

We consider four different click strategies to emulate a range of user click behaviours:

\begin{description}[style=unboxed,leftmargin=0cm]
	\item [``No Click''] No items are clicked on in the response page to a query. This user click-model does not provide additional user preference information to the recommender system due to click behaviour. This click model is used in the baseline measurements presented in Sections~\ref{sec:test:setup}.
	\item [``Click Relevant''] Given the response page to a query, for each search result and advert we calculate the Term-Frequency (TF) of the visible text with respect to the keywords associated with the test session topic of interest. When $TF > 0.1$ for an item, the item is clicked, otherwise it is not clicked. This user click-model provides relevant feedback to the recommender system about the information goal of the user.
	\item [``Click Non-relevant''] TF is calculated for each item with respect to the category of interest for the session in question as for the ``Click Relevant'' click-model, \emph{except} that items are clicked when the TF score is below the threshold and so they are deemed non-relevant to the topic, that is when $TF \leq 0.1$. This user click-model attempts to confuse the recommender system by providing feedback that is not relevant to the true topic of interest to the user.
	\item [``Click All''] All items on the response page for a query are clicked. This user click-model gives the recommender system a ``noisy'' click signal, including clicks on items relevant and non-relevant to the user's information goal.
	\item [``Click 2 Random Items''] Two items appearing on the response page for a query are selected uniformly at random with replacement and clicked.  
\end{description} 

In all cases, when uninteresting, noise queries are included in a query session, the relevant user click-strategy is also applied to the result pages of these queries. In this way we hope to avoid providing an obvious signal to the recommender system that might differentiate uninteresting queries from queries related to  sensitive topics.  Items on the result page in response to probe queries are \emph{not} clicked so that the probe query does not provide any additional information to the recommender system.

\begin{table}[h]
\captionsetup{position=bottom,skip=0pt,belowskip=0pt}
\caption{Measured Plausible Deniability versus any other tested topics as probability of interest, by Probe Query Sequence when the true topic of interest is ``Other'' with range $(\mu \pm 3\sigma)$\label{tbl:base:allclicks:nonoise}}
\setlength{\tabcolsep}{0.30cm}
\begin{scriptsize}

\subfloat[Click Relevant, No Noise]{

\begin{tabular}{@{}lccccc@{}}
\textbf{\shortstack[l]{Reference\\Topic\\{}}}   & \textbf{Probe 1}  & \textbf{Probe 2}   & \textbf{Probe 3}  & \textbf{Probe 4} & \textbf{Probe 5}  \\ 
\cmidrule[0.5pt]{1-6}
\textbf{anorexia} & 59 (50) & 59 (50) & 59 (50) & 59 (50) & 59 (50) \\
\textbf{bankrupt} & 16 ( 8) & 65 (42) & 65 (36) & 59 (40) & 54 (38) \\
\textbf{diabetes} & 36 (36) & 36 (36) & 36 (36) & 36 (36) & 36 (36) \\
\textbf{disabled} &  7 ( 4) &  7 ( 4) &  9 ( 9) & 40 ( 4) & 40 ( 7) \\
\textbf{divorce} & 30 (24) & 30 ( 9) & 30 ( 9) & 30 ( 8) &  30 ( 8) \\
\textbf{gambling} &  6 ( 0) & 18 (16) & 32 (16) & 18 (16) & 18 ( 5) \\
\textbf{gay} & 92 (51) & 92 (77) & 78 (51) & 94 (72) & 94 (80) \\
\textbf{location} & 18 (18) & 10 (10) & 10 (10) & 18 (10) & 18 (10) \\
\textbf{payday} &  2 ( 2) &  2 ( 2) &  3 ( 2) &  3 ( 2) &  2 ( 2) \\
\textbf{prostate} & 17 (17) & 17 (17) & 17 (17) & 17 (17) & 17 (17) \\
\textbf{unemployed} & 13 ( 2) & 13 ( 4) & 13 ( 7) & 13 ( 7) &  7 ( 6) \\
 \cmidrule[0.5pt]{1-6}                
 \end{tabular}        
} \\[-3mm]
\subfloat[Click Non-relevant, No Noise]{

\begin{tabular}{@{}lccccc@{}}
\textbf{\shortstack[l]{Reference\\Topic\\{}}}   & \textbf{Probe 1}  & \textbf{Probe 2}   & \textbf{Probe 3}  & \textbf{Probe 4} & \textbf{Probe 5}  \\ 
\cmidrule[0.5pt]{1-6}
\textbf{anorexia} & 18 ( 5) & 22 (12) & 26 ( 5) & 31 (13) & 32 ( 6) \\
\textbf{bankrupt} & 57 ( 3) & 53 (36) & 50 (34) & 43 (33) & 48 (36) \\
\textbf{diabetes} &  4 ( 2) & 13 ( 8) & 11 ( 8) &  5 ( 3) & 11 ( 2) \\
\textbf{disabled} &  5 ( 2) &  6 ( 2) &  9 ( 3) & 29 (10) & 26 ( 8) \\
\textbf{divorce} & 49 (25) & 51 (33) & 49 (30) & 43 (29) &  43 (29) \\
\textbf{gambling} &  6 ( 2) & 18 ( 4) & 36 (24) & 35 (13) & 31 (13) \\
\textbf{gay} & 36 (33) & 75 (33) & 51 (32) & 39 (20) & 31 (27) \\
\textbf{location} &  9 ( 2) & 11 ( 1) &  7 ( 2) &  6 ( 2) &  9 ( 1) \\
\textbf{payday} &  3 ( 3) &  3 ( 1) &  4 ( 2) &  3 ( 2) &  4 ( 3) \\
\textbf{prostate} & 55 (38) & 68 (36) & 65 (48) & 61 (48) & 64 (42) \\
\textbf{unemployed} &  9 ( 1) &  6 ( 6) &  7 ( 1) &  9 ( 4) &  5 ( 2) \\
 \cmidrule[0.5pt]{1-6}                
 \end{tabular}        
} \\[-3mm]
\subfloat[Click All, No Noise]{

\begin{tabular}{@{}lccccc@{}}
\textbf{\shortstack[l]{Reference\\Topic\\{}}}   & \textbf{Probe 1}  & \textbf{Probe 2}   & \textbf{Probe 3}  & \textbf{Probe 4} & \textbf{Probe 5}  \\ 
\cmidrule[0.5pt]{1-6}
\textbf{anorexia} & 66 (57) & 66 (57) & 66 (57) & 66 (57) & 66 (57) \\
\textbf{bankrupt} &  51 (42) & 51 (42) & 51 (42) & 55 (46) & 56 (46) \\
\textbf{diabetes} & 35 (35) & 35 (35) & 35 (35) & 35 (35) & 35 (35) \\
\textbf{disabled} &  9 ( 9) &  9 ( 9) &  9 ( 9) & 31 (31) & 31 (31) \\
\textbf{divorce} & 30 ( 8) & 73 (54) & 54 (34) & 100 (49) &  100 (49) \\
\textbf{gambling} &  3 ( 1) & 16 (16) & 53 (11) & 16 ( 6) &  6 ( 2) \\
\textbf{gay} & 69 (65) & 77 (73) & 70 (60) & 82 (75) & 81 (71) \\
\textbf{location} & 18 (10) & 10 ( 6) & 10 ( 6) & 14 (10) & 18 ( 7) \\
\textbf{payday} &  2 ( 2) &  2 ( 2) &  2 ( 2) &  2 ( 2) &  2 ( 2) \\
\textbf{prostate} & 17 (17) & 17 (17) & 17 (17) & 17 (17) & 17 (17) \\
\textbf{unemployed} &  4 ( 4) &  7 ( 7) &  7 ( 7) &  7 ( 7) &  7 ( 6) \\
 \cmidrule[0.5pt]{1-6}                
 \end{tabular}        
} \\[-3mm]
\subfloat[Click 2 Random Items, No Noise]{

\begin{tabular}{@{}lccccc@{}}
\textbf{\shortstack[l]{Reference\\Topic\\{}}}   & \textbf{Probe 1}  & \textbf{Probe 2}   & \textbf{Probe 3}  & \textbf{Probe 4} & \textbf{Probe 5}  \\ 
\cmidrule[0.5pt]{1-6}
\textbf{anorexia} & 50 (12) & 27 ( 9) & 26 ( 9) & 36 (10) & 33 (11) \\
\textbf{bankrupt} &  5 ( 3) & 43 (33) & 39 (37) & 36 (35) & 38 (35) \\
\textbf{diabetes} & 38 ( 6) & 18 ( 7) & 17 ( 5) & 17 ( 7) & 11 ( 5) \\
\textbf{disabled} &  2 ( 1) &  4 ( 1) &  5 ( 3) & 39 (25) & 40 (25) \\
\textbf{divorce} & 24 (17) & 37 (31) & 37 (31) & 35 (25) &  35 (25) \\
\textbf{gambling} & 24 ( 0) &  7 ( 4) & 54 (23) & 33 (23) & 68 (20) \\
\textbf{gay} & 68 (68) & 68 (65) & 54 (52) & 46 (36) & 47 (42) \\
\textbf{location} &  8 ( 8) &  8 ( 8) &  8 ( 8) &  8 ( 8) &  8 ( 8) \\
\textbf{payday} &  4 ( 1) &  2 ( 2) &  4 ( 2) &  4 ( 3) &  4 ( 4) \\
\textbf{prostate} & 59 (57) & 67 (62) & 58 (56) & 60 (54) & 51 (44) \\
\textbf{unemployed} &  4 ( 3) &  8 ( 3) & 10 ( 4) &  3 ( 2) & 10 ( 1) \\
 \cmidrule[0.5pt]{1-6}                
 \end{tabular}        
} \\[-6mm]                                     
\end{scriptsize}
\end{table}%

Measure values of \PDE{} are shown in Table~\ref{tbl:base:allclicks:nonoise}. As random noise injection had no observable effect on measurements of \PDE{} for different click models in experiments, only the ``No Noise'' results are presented here for space reasons.

Taken overall, the results in Table~\ref{tbl:base:allclicks:nonoise}(a) for the ``non-relevant click, no noise'' model suggest clicking on non-relevant advert items is the best strategy of the click models tested. The only difference between the ``non-relevant click'' model and other click models  is that non-relevant items \emph{only} are clicked, whereas in other click models it is possible that relevant items are clicked. It seems reasonable to postulate that clicking on relevant items provides ``fine-tuned'' feedback about user interests which is more informative for the observer. Clicking on non-relevant items may divert attention to a modest degree, but not to the extent of masking the sensitive topic revealed by the query. 

Comparing the baseline ``No Click'' \PDE{} observations in Table~\ref{tbl:base:noclick:nonoise} each of the subtables in Table~\ref{tbl:base:allclicks:nonoise} shows similar lack of consistency to the noise injection models. In out experiments there is no consistent change observed in \PDE{} across topics due to variation in the click patterns tested. As with the noise injection case, there are sporadic increases and decreases in values of \PDE{} but the lack of overall consistency makes using click models as a defence impractical. 

It would appear in summary, that clicks transmit information to the observer, but not as consistently as does a revealing query. Consequently none of the user click-models tested appear to change the baseline level of plausible deniability associated with the query in a predictable way so that there is no globally discernible pattern with which to construct practical defence tools based on clicks.

\subsection{The Effect of Proxy Topics}
\label{sec:proxy:topic}

The next privacy protection strategy we consider is the introduction of proxy topics. In this case sequences of queries, with each sequence related to a single proxy topic which is not sensitive for the user but capable of attracting personalised advert content, are injected into a user session.   The idea here is that each such sequence of queries emulates a user session where the proxy topic is the topic of interest. In this way we hope to misdirect learning by the search engine of user interests.  The results in Section \ref{sec:testing:noise} are relevant here since they suggest that isolated, individual queries -- such as randomly selected noise queries -- tend not to provoke search engine learning. Our hope is that this can be exploited by inverting the notion of random noise injection so that individual \emph{sensitive} queries are injected as the noise in proxy topic sessions. Isolated sensitive queries will hopefully not provoke learning whereas the larger number of uninteresting proxy sessions will. In this way we can misdirect learning by the observer.

In out tests the following proxy topics are used:
\begin{description}[style=unboxed,leftmargin=0cm]
	\item[\textbf{Tickets}] Searching for tickets for events in a well-known local stadium
	\item[\textbf{Vacation}] Queries related to a vacation such as flights and accommodation.
	\item[\textbf{Car}] Searches by a user seeking to trade in and change their car.
\end{description}
and related queries are constructed by selecting related keywords through the same process as was used for the sensitive topics.

\begin{table}[h]
\captionsetup{position=bottom,skip=0pt,belowskip=0pt}
\caption{Measured Plausible Deniability versus any other tested topics as probability of interest, by Probe Query Sequence when the true topic of interest is ``Other'' with range $(\mu \pm 3\sigma)$\label{tbl:base:proxy:topics}}
\setlength{\tabcolsep}{0.30cm}
\begin{scriptsize}

\subfloat[All Click and Noise Models]{

\begin{tabular}{@{}lccccc@{}}
\textbf{\shortstack[l]{Reference\\Topic\\{}}}   & \textbf{Probe 1}  & \textbf{Probe 2}   & \textbf{Probe 3}  & \textbf{Probe 4} & \textbf{Probe 5}  \\ 
\cmidrule[0.5pt]{1-6}
\textbf{all topics} &  0 ( 0) &  0 ( 0) &  0 ( 0) &  0 ( 0) &  0 ( 0) \\
%
 \cmidrule[0.5pt]{1-6}                
 \end{tabular}        
} \\[-3mm]
\end{scriptsize}
\end{table}%
Proxy topic query scripts where constructed by selecting a sensitive topic, and then selecting an uninteresting proxy topic from the list of $3$ proxy topics. Having decided on a sensitive query we wish to issue, we select at least three and no more than four queries related to the proxy topic from a prepared list of proxy topic queries. We next randomly shuffle the order of the selected sensitive and proxy topic queries. In this way there is always a subgroup of at least two proxy topic queries next to each other in each query session. Finally, for testing purposes, we place a probe query before and after each block of 3-4 proxy + 1 sensitive queries to measure changes in \PRIEPS{} score. We repeat this exercise using the same proxy topic until a typical query session consisting of $5$ probe queries is created. 

Data was collected for $2,300$ such proxy topic sessions. This included each  of the sensitive topics and each of the click models described in previous sections. The same \PRIEPS{} and \PDE{} setup, including the same training set, as before was used to process the search results.  

Measured detection rates are shown in Table~\ref{tbl:base:proxy:topics}. The measured probability calculated from \PDE{} is $0$ for all topics and for all click-models tested.  That is, we find it is possible to claim full plausible deniability of interest in all of the topics tested. Since our detection approach is demonstrated to be notably sensitive to observer learning in earlier sections, we can reasonably infer that this result is not due to a defect in the detection methodology but rather genuinely reflects successful misdirection of the search engine away from sensitive topics. 

This result is encouraging, especially in light of the negative results in previous sections for other obfuscation approaches.  It suggests use of sequences of queries on uninteresting proxy topics may provide a defence of plausible deniability. The trade-offs for the user include the overhead of maintaining proxy topics and associated queries and the additional resources required to issue proxy topic queries in a consistent way. However since both of these tasks were readily automated during our testing it seems reasonable that these trade-offs could be readily managed by software in a way that is essentially transparent to the user.
\section{Conclusions and Discussion}
\label{sec:conclusions}
Our observations suggest that modern systems, such as Google, are able to identify user interests with high accuracy, exploit multiple signals, filter out uninteresting noise queries and adapt quickly when topics change. Furthermore learning appears to be sustained over the lifetime of query sessions. The power and sophistication of these systems make designing a robust defence of user privacy non-trivial.

The \PDE{} estimator was tested via a comprehensive measurement program using online search engines to show that topic learning results in measurable impacts on the ability of a user to deny their interest in all sensitive topics tested. We find that revealing queries provide a significant signal for search engine adaptation. While user clicks provide additional feedback, we do not observe the same degree of associated learning with click behaviour as is observed with revealing queries. Overall, testing with \PDE{} suggests that defences based on random noise injection and variable click models do not provide a reliable strategy for defence of plausible deniability. 

By contrast, our experiments show that proxy topics that are uninteresting to the user but capable of generating commercial content provide observable privacy protection in our experiments. Wrapping sensitive queries in a stream of coherent proxy topic queries appears to distract the online system into adapting to the proxy topic while allowing the sensitive query noise to slip through. Our observation that proxy topics provide some relief indicates that defence of plausible deniability is not impossible, but indicates that increasingly sophisticated approaches are required in the face of ever improving search engine capability. In choosing proxy topics, for example, a user must be careful to not stimulate unintended learning of the proxy topics which may influence the utility of future search results.

Subtle tactics like proxy topics, that exploit the observer's strengths to tip the balance slightly in favour of the user, suggest an interesting avenue for future research. The simplicity of the approach means it should be possible to extend it in several ways, for example, by injecting a range of uninteresting single topic queries as additional noise in the proxy query stream it may be possible to provide additional guarantees of privacy such as \emph{k-anonymity} or \emph{differential privacy} for the sensitive topic. \POL{More investigation of proxy topics is an interesting line of future research. Experiments to compare the effectiveness of different proxy topics including, for example, inclusion of proxy topics that are more relevant to the user's known interests versus proxy topics that are less relevant to user topics. Similarly proxy topics with higher commercial value may have more potential to distract search engine learning than proxy topics with lower commercial value}

As discussed in Section~\ref{sec:related:work}, user click patterns may be used by recommender systems to rank page content, placing content likely to attract user clicks in more prominent positions on pages. In our experiments, we observed changes in volume of advert content on samples of probe query response pages. There are several plausible avenues of investigation that may help explain the mechanism behind this, such as user click patterns and the semantics of the true and noise queries chosen. The approach taken in this paper does not distinguish between items based on rank or order on the page.  How the semantics of queries, the interaction between user click-models and the effect of content ranking may impact user privacy is beyond the scope of this current paper and an avenue for future research.

Overall our results point towards an arms race, where search engine capability is continuously evolving. In this setting, even if injection of proxy topic sessions were to become widely deployed then we can reasonably expect search engines to respond with more sophisticated learning strategies.  Our results also point towards the fact that the text in search queries plays a key role in search engine learning. While perhaps obvious, this observation reinforces the user's need to be circumspect about the queries that they ask if they want to avoid search engine learning of their interests.

\bibliographystyle{unsrt}
\begin{scriptsize}
\bibliography{noisy}
\end{scriptsize}
%
\appendices
\section{Additional Results}
\label{sec:additional:1}
\small
\begin{lemma}
\label{lemma:1}
For $x,y,\epsilon \in \RPLUS$ with $0<x,y<1$ 
\begin{align}
e^{-\epsilon} < \frac{x}{y} < e^{\epsilon} \implies \vert x-y \vert < \epsilon
\label{eqn:lemma:000}	
\end{align}	
\end{lemma}
\begin{proof}
Assuming the left hand side of \eqref{eqn:lemma:000} holds
\begin{align*}
&e^{-\epsilon} < \frac{x}{y} < e^{\epsilon}
\iff ye^{-\epsilon} < x \text{ and } y > xe^{-\epsilon}
\\
&\implies y(1-\epsilon) < x \text{ and } y > x(1-\epsilon)
\text{ (Since $e^{-x} > 1-x$)} \\
&\iff y-x<y\epsilon \quad\text{and}\quad x-y < x\epsilon
\\
&\iff y-x<\epsilon \quad\text{and}\quad x-y < \epsilon
\text{ (Since $x,y < 1$)} \notag \\
&\iff -\epsilon < x-y < \epsilon
\iff \vert x-y \vert < \epsilon
\end{align*}%
\QED  
\end{proof}%
\normalsize
\begin{table}[h]
\captionsetup{position=bottom,skip=0pt,belowskip=0pt}
\caption{Comparison of measured detection rate of \emph{at least one} individual ``Sensitive' topic in a session of $5$ probes.\label{tbl:old:new}}
\setlength{\tabcolsep}{0.20cm}
\begin{scriptsize}
\subfloat[\PRIEPS]{
\begin{tabular}{@{}lc@{}}
\textbf{\shortstack[l]{Reference\\Topic}}   & \textbf{\shortstack[l]{All\\Topics}}  \\ \cmidrule[0.5pt]{1-2}
    \multicolumn{1}{l}{\textbf{True Detect}}      & 100.0\%  \\
    \multicolumn{1}{l}{\textbf{False Detect}}     & 0.0\%    \\ 
    \cmidrule[0.5pt]{1-2}
\end{tabular}        
} \hfill 
\subfloat[\PRI]{

\begin{tabular}{@{}lc@{}}
\textbf{\shortstack[l]{Reference\\Topic}}   & \textbf{\shortstack[l]{All\\Topics}}  \\ \cmidrule[0.5pt]{1-2}
    \multicolumn{1}{l}{\textbf{True Detect}}      & 97-100.0\%  \\
    \multicolumn{1}{l}{\textbf{False Detect}}     & 4-8\%    \\ 
    \cmidrule[0.5pt]{1-2}
\end{tabular} 
} \\[-6mm]               
\end{scriptsize}
\end{table}%
\end{document}